\documentclass[3p]{elsarticle}

\makeatletter
\def\ps@pprintTitle{%
 \let\@oddhead\@empty
 \let\@evenhead\@empty
 \def\@oddfoot{\centerline{\thepage}}%
 \let\@evenfoot\@oddfoot}
\makeatother

\usepackage{amsthm}
\usepackage{amssymb}
\usepackage{amsfonts}
\usepackage{verbatim}
\usepackage{amsmath}
\usepackage{amsfonts}
\usepackage[english]{babel}
\usepackage{amscd}
\usepackage{mathptmx}
\usepackage{mathrsfs}
\usepackage{color}
\usepackage{xspace}

\theoremstyle{plain} 
\newtheorem{theorem}{Theorem}[section] 
\newtheorem{lem}[theorem]{Lemma}
\newtheorem{prop}[theorem]{Proposition}
\newtheorem{defn}[theorem]{Definition}

\newcommand{\logicFont}[1]{\mathsf{#1}\xspace}

\usepackage[cal=boondox,scr=boondoxo]{mathalfa}

\usepackage{hyperref}

\usepackage{float}

\usepackage{bussproofs}

\usepackage{multirow}
\usepackage{cleveref}

\usepackage{enumerate}
\usepackage{enumitem}

\DeclareSymbolFont{letters}{OML}{cmm}{m}{it}

\DeclareMathAlphabet{\mathcal}{OMS}{cmsy}{m}{n}

\usepackage{relsize}

\newcommand{\FO}{\logicFont{FO}}
\newcommand{\ESO}{\logicFont{ESO}}

\newcommand{\FOT}{\ensuremath{\logicFont{FOT}}\xspace}

\newcommand{\FOPTinc}{\ensuremath{\logicFont{FOPT}(\leq^\delta)}\xspace}
\newcommand{\FOPTinccond}{\ensuremath{\logicFont{FOPT}(\leq^\delta,\cisym\hspace{-.8mm}^\delta_\mathrm{c})}\xspace}
\newcommand{\FOPTcondineq}{\ensuremath{\logicFont{FOPT}(\leq_\mathrm{c}^{\delta})}\xspace}

\newcommand{\FOrxsum}{\ensuremath{\logicFont{FO}_{\mathbb{R}_{\geq 0}}(\times,\SUM)}\xspace}
\newcommand{\FOrsum}{\ensuremath{\logicFont{FO}_{\mathbb{R}_{\geq 0}}(\SUM^*)}\xspace}

\newcommand{\existso}{\ensuremath{\exists^1}\xspace}
\newcommand{\forallo}{\ensuremath{\forall^1}\xspace}
\newcommand{\wcn}{\ensuremath{\mathop{\dot\sim}}\xspace}

\newcommand{\dblsetminus}{\mathbin{{\setminus}\mspace{-5mu}{\setminus}}}
\newcommand{\vvee}{\raisebox{1pt}{\ensuremath{\,\mathop{\mathsmaller{\mathsmaller{\dblsetminus\hspace{-0.23ex}/}}}\,}}}

\newcommand{\supp}{\ensuremath{\textnormal{supp}}\xspace}
\newcommand{\distr}{\ensuremath{\textnormal{distr}}\xspace}\newcommand{\ddfn}{::=}

\newcommand{\ci}[3]{#2~\bot_{#1}~#3}

\newcommand{\pci}[3]{{#2 \perp\!\!\!\!\perp_{#1} #3}}
\newcommand{\cisym}{\perp\!\!\!\!\perp}

\newcommand{\Dom}{\mathrm{Dom}}

\newcommand{\SUM}{\mathrm{SUM}}

\newcommand{\X}{\mathbb{X}}

\newcommand{\mA}{\mathcal{A}}

\newcommand{\Var}{\mathrm{Var}}

\newcommand{\ar}{\mathrm{ar}}

\usepackage{hyperref}

\makeatletter
\providecommand{\doi}[1]{%
  \begingroup
    \let\bibinfo\@secondoftwo
    \urlstyle{rm}%
    \href{http://dx.doi.org/#1}{%
      doi:\discretionary{}{}{}%
      \nolinkurl{#1}%
    }%
  \endgroup
}
\makeatother

\makeatletter
\def\@author#1{\g@addto@macro\elsauthors{\normalsize%
    \def\baselinestretch{1}%
    \upshape\authorsep#1\unskip\textsuperscript{%
      \ifx\@fnmark\@empty\else\unskip\sep\@fnmark\let\sep=,\fi
      \ifx\@corref\@empty\else\unskip\sep\@corref\let\sep=,\fi
      }%
    \def\authorsep{\unskip,\space}%
    \global\let\@fnmark\@empty
    \global\let\@corref\@empty  
    \global\let\sep\@empty}%
    \@eadauthor={#1}
}
\makeatother

\begin{document}

\begin{frontmatter}{}

\author{Miika Hannula\fnref{fn1}}
\author{Minna Hirvonen\corref{cor1}\fnref{fn2}}
\ead{minna.hirvonen@helsinki.fi}
\author{Juha Kontinen\fnref{fn1}}
\address{Department of Mathematics and Statistics, University of Helsinki, Helsinki, Finland}

\cortext[cor1]{Corresponding author}
\fntext[fn1]{Supported by  grant 308712 of the Academy of Finland.}
\fntext[fn2]{Supported by the Vilho, Yrj\"o and Kalle V\"ais\"al\"a Foundation.} 
\title{On elementary logics for quantitative dependencies}

\begin{abstract}
We define and study  logics in the  framework of probabilistic team semantics and over metafinite structures. Our work is paralleled by the recent development of novel axiomatizable and tractable logics in team semantics that are closed under the Boolean negation. Our logics employ new  probabilistic atoms that resemble so-called extended atoms from the team semantics literature. We also define counterparts of our logics over metafinite structures and show that all of our logics can be translated into functional fixed point logic implying a polynomial time upper bound for data complexity with respect to BSS-computations.
\end{abstract}


\end{frontmatter}{}

\section{Introduction}


In this article we define new logics for the framework of \emph{probabilistic team semantics}. Our work is motivated and paralleled by the recent development of novel axiomatizable and tractable logics in team semantics that are closed under the Boolean negation but remain on the first-order level for expressivity. Our logics employ new probabilistic atoms that resemble so-called extended atoms from the team semantics literature. Unlike, e.g., with extended dependence atoms which are definable by usual dependence atoms, the new extended quantitative dependencies are a crucial feature of the new logics. We also define counterparts of our logics over metafinite structures and show that all of our logics can be translated into functional fixed point logic giving a deterministic polynomial-time upper bound for data complexity with respect to BSS-computations.

Team semantics is a semantical framework originally introduced by Hodges \cite{Hodges1997a}   and  V\"a\"an\"anen with the introduction of {\em dependence logic} \cite{Van07dl}. Soon after the introduction of dependence logic, the focus in (first-order) team semantics turned to independence logic  and inclusion logic that were introduced in \cite{D_Ind_GV,Pietro_I/E}.  During the past decade  research on logics in team semantics has flourished with interesting connections to many fields such as database theory  \cite{HannulaK16}, statistics \cite{CoranderHKPV16}, and temporal hyperproperties \cite{kmvz18}.

In team semantics formulas are evaluated over  sets of assignments (called {\em teams}) rather than single assignments as in first-order logic. This feature has the effect that knowing the expressive power of a logic for sentences does not immediately give a characterization for the expressive power of the open formulas of the logic. For example, while it follows  from the earlier results of \cite{henkin61,Enderton1970,Walkoe1970} that dependence logic  and independence logic  are both equivalent to existential second-order logic ($\ESO$) on the level of sentences, the open formulas of dependence logic are strictly less expressive compared to  independence logic: The latter characterizes all $\ESO$-definable team properties \cite{Pietro_I/E}, whereas the former only downward closed $\ESO$-definable properties \cite{KontVan09}. 

A salient feature of (most) logics in team semantics is that their expressive power exceeds that of first-order logic.  Only recently  a  team-based logic  \FOT was defined whose expressive power coincides with first-order logic  both on the level of sentences and  open formulas. Previously it had been observed e.g.,  that the extensions of $\FO$ by constancy atoms or the Boolean negation $\sim$ are both equivalent to $\FO$ over sentences but strictly less expressive than $\FO$ for open  formulas  when the team is represented by a relation \cite{Galliani2016,Luck18}.  The logic \FOT utilizes a weaker version of disjunction and the existential quantifier in order not to go beyond the expressivity of $\FO$ (see \cite{DurandKRV15} for a systematic study of this phenomenon). We will follow the same strategy when defining our new logics in the probabilistic setting.

In this paper our focus is on  probabilistic team semantics that extends the area of team semantics from qualitative to quantitative dependencies such as probabilistic  independence.  A probabilistic team is a set of assignments with an additional function that maps each assignment to some numerical value. Usually, the function is a probability distribution, but it can also be thought of as a frequency distribution. We allow the values to be any non-negative real numbers.  The systematic study of logics in probabilistic team semantics was initiated by the works \cite{DurandHKMV18,HKMV18} and  they have already found applications, e.g.,  in the study of the implication problem of conditional independence \cite{jelia19} and the  foundations of quantum mechanics \cite{G21}.

By the results of  \cite{HKMV18,real20}  probabilistic independence logic 
 is equivalent to a sublogic of $\ESO$ interpreted over so-called $\mathbb{R}$-structures  ($\ESO_{\mathbb{R}}$). 
In this paper our goal is to initiate a study of  tractable probabilistic logics  and to find their analogues over metafinite structures. We note that the tractability frontier of  the previously defined logics in probabilistic team semantics  has been recently charted in  \cite{MJ21}.
We introduce a new logic called $\FOPTinccond$, in which the disjunction and the quantifiers are similar to the ones in $\FOT$ and the atoms compare the probabilities of events defined by quantifier-free formulas. In fact, the logic $\FOPTinccond$ can be seen as a generalization of  $\FOT$ 
for  probabilistic team semantics. In addition to the qualitative atoms expressible in  $\FOT$, certain previously studied probabilistic atoms, i.e. marginal identity and probabilistic conditional independence, are also expressible in $\FOPTinccond$. 

We also define two other team-based logics: $\FOPTinc$ which is  a fragment of $\FOPTinccond$, and $\FOPTcondineq$ in which every formula of $\FOPTinccond$ is expressible. The logic $\FOPTcondineq$ features a new type of atom, \textit{conditional probability inequality}, that can be used to compare conditional probabilities. With this atom, we can express both kinds of extended atoms from $\FOPTinccond$, i.e. the extended probabilistic inclusion and the extended probabilistic conditional independence. We also take a look at $\FOPTcondineq$ from a complexity theoretic point of view and show that its satisfiability and validity problems are \textsc{RE}-complete and co-\textsc{RE}-complete, respectively. 
 

In the second part of the article we consider logics over  two-sorted (metafinite) structures which, in addition to a finite structure, come with an infinite second sort and functions that bridge the two sorts. We define a logic, $\FOrxsum$, which is an extension of first-order logic on 
metafinite structures with a numerical second sort that has  access to multiplication and aggregate sums over non-negative reals. We show that $\FOPTinccond$ can be translated into $\FOrxsum$, and identify a fragment of $\FOrxsum$ which is equi-expressive with $\FOPTinc$. We also give a translation from $\FOrxsum$ to functional fixed point logic $\logicFont{FFP}_{\mathbb{R}}$ over metafinite structures 
and thus obtain a polynomial time upper bound for the data complexity of our new logics in the BSS-model. 

\section{Preliminaries}

First-order variables are denoted by $x,y,z$ and tuples of first-order variables by $\bar{x},\bar{y},\bar{z}$. The set of variables that appear in the tuple $\bar{x}$ is denoted by $\Var(\bar{x})$, and by $|\bar{x}|$, we denote the length of the tuple $\bar{x}$. A vocabulary $\tau$ is a finite set of relation, function, and constant symbols, denoted by $R$, $f$, and $c$, respectively. Each relation symbol $R$ and function symbol $f$ has a prescribed arity which we denote by $\ar(R)$ and $\ar(f)$. 

A vocabulary $\tau$ is called \textit{relational} if it only contains relation symbols, and \textit{functional} if it only contains function symbols. We sometimes assume that the vocabulary we are considering is relational. This assumption can be made without loss of generality since each function can be expressed by a relation that describes its graph. For some proofs, it is useful to allow the vocabulary to contain constants, and therefore we sometimes assume that the vocabulary solely consists of relation and constant symbols.

\subsection{Team semantics and the logics $\FOT$ and $\FOT^{\downarrow}$}

Let $\tau$ be a finite
vocabulary that only contains relation and constant symbols. We assume that $\{=\}\subseteq\tau$. Let $D$ be a finite set of variables and $\mA$ a finite $\tau$-structure. An \textit{assignment} of a structure $\mA$ for the set $D$ is a function $s\colon D\to A$. A \textit{team} $X$ of $\mA$ over the set $D$ is a finite set of assignments $s\colon D\to A$\footnote{Note that unlike in our version of probabilistic team semantics, here $X$ is not required to be maximal; it can be any finite set of assignments.}. The set $D$ is also called the \textit{domain} of $X$, or $\Dom(X)$ for short. For a variable $x$ and $a\in A$, we denote by $s(a/x)$, the modified assignment $s(a/x)\colon D\cup\{x\}\to A$ such that $s(a/x)(y)=a$ if $y=a$, and $s(a/x)(y)=s(y)$ otherwise. The modified team $X(a/x)$ is defined as the set $X(a/x):=\{s(a/x)\mid s\in X\}$.

We consider two team-based logics, $\FOT$ and $\FOT^{\downarrow}$, which were introduced in \cite{fot19}. The expressive power of $\FOT$ coincides with first-order logic, and $\FOT^{\downarrow}$ captures downward closed first-order team properties \cite{fot19}. The logics that we introduce in section \ref{FOPT} can be seen as generalizations of these two logics. 

First-order $\tau$-terms and atomic formulas are defined in the usual way. We let 
\begin{equation}\label{delta}
\delta\ddfn\lambda \mid \neg\delta \mid \delta\wedge\delta
\end{equation}
for any first-order atomic formula $\lambda$ of the vocabulary $\tau$. Let $x$ be a first-order variable, and let $\bar{x}$ and $\bar{y}$ be tuples of variables with $|\bar{x}|=|\bar{y}|$. The logic $\FOT$ over a vocabulary $\tau$ is then defined as follows:
\[
\phi\ddfn\lambda\mid \bar{x}\subseteq\bar{y}\mid \wcn\phi \mid \phi\wedge\phi \mid \phi\vvee\phi \mid \existso x\phi \mid \forallo x\phi,
\]
and the logic $\FOT^{\downarrow}$ as follows:
\[
\phi\ddfn\delta \mid \phi\wedge\phi \mid \phi\vvee\phi \mid \existso x\phi \mid \forallo x\phi.
\]
Note that even though $\FOT$ does not contain the negation symbol $\neg$, the formula $\neg\delta$ 
is expressible in $\FOT$ using $\subseteq$, $\wcn$, and $\vvee$, as shown in \cite{fot19}.

The semantics for the two logics is defined as follows:
\begin{itemize}
\item $\mA\models_{X}\delta$ iff $\mA\models_{s}\delta$ for all $s\in X$.
\item $\mA\models_{X}\bar{x}\subseteq\bar{y}$ iff for all $s\in X$, there exists $s'\in  X$ such that $s(\bar{x})=s'(\bar{y})$.
\item $\mA\models_{X}\wcn\phi$ iff $\mA\not\models_{X}\phi$ or $X=\varnothing$.
\item $\mA\models_{X}\phi\wedge\psi$ iff $\mA\models_{X}\phi$ and $\mA\models_{X}\psi$.
\item $\mA\models_{X}\phi\vvee\psi$ iff $\mA\models_{X}\phi$ or $\mA\models_{X}\psi$.
\item $\mA\models_{X}\existso x\phi$ iff $\mA\models_{X(a/x)}\phi$ for some $a\in A$. 
\item $\mA\models_{X}\forallo x\phi$ iff $\mA\models_{X(a/x)}\phi$ for all $a\in A$.
\end{itemize}
Note that if $X$ is empty, then $\mA\models_{X}\phi$ for any $\phi\in\FOT[\tau]$ or $\phi\in\FOT^{\downarrow}[\tau]$.

\subsection{Probabilistic team semantics}

Let $\tau$, $D$, $\mA$, and $X$ be as above, with the exception that we assume that $X$ is maximal, i.e. it contains all assignments $s\colon D\to A$. A \textit{probabilistic team} $\X$ is a function $\X\colon X\to\mathbb{R}_{\geq 0}$, where $\mathbb{R}_{\geq 0}$ is the set of non-negative real numbers. The value $\X(s)$ is also called the \textit{weight} of assignment $s$. We define the \textit{support} of $\X$ as follows:
\[
\supp(\X):=\{s\in X\mid\X(s)\neq 0\},
\]
and say that the team $\X$ is \textit{nonempty} if $\supp(\X)\neq\varnothing$. Note that even when $D=\varnothing$, the probabilistic team $\X$ may still be nonempty: if $D=\varnothing$, then $X$ is the singleton containing the empty assignment whose weight can be set as nonzero.

Functions $\X\colon X\to\mathbb{R}_{\geq 0}$ such that $\sum_{s\in X}\X(s)=1$ are called \textit{probability distributions}. They are an important special case of probabilistic teams and originally probabilistic teams were required to be probability distributions (hence the name \textit{probabilistic} team). If $\X$ is a probability distribution, we also write $\X\colon X\to [0,1]$. Note that from every nonempty probabilistic team $\X\colon X\to\mathbb{R}_{\geq 0}$ team we obtain a probability distribution $\distr(\X)\colon X\to [0,1]$ by setting 
\[
\distr(\X)(s)=\frac{1}{\sum_{t\in X}\X(t)}\cdot\X(s) 
\]
for all $s\colon D\to A$. It does not matter whether we evaluate formulas using the original team or the team that has been scaled in order to obtain a probability distribution 
(see Proposition \ref{distrprop}).

By $\X(a/x)$, we denote the probabilistic team such that 
\[
\X(a/x)(s)=\sum_{\substack{t\in X,\\ t(a/x)=s}}\X(t)
\]
for all $s\colon D\cup\{x\}\to A$. Note that if $x$ is a fresh variable (i.e. $x\notin D$), then for all $s\in X$,
\[
\X(a/x)(s(b/x))=
\begin{cases}
\X(s), &\text{ when } b=a\\
0, &\text{ when } b\neq a.
\end{cases}
\]

\section{Logics in probabilistic team semantics}\label{FOPT}


\subsection{The logics $\FOPTinccond$ and $\FOPTinc$}

First-order $\tau$-terms and atomic formulas are defined in the usual way. Let $\delta$ be as in Equation \ref{delta}. The logic $\FOPTinccond$ over a vocabulary $\tau$ is then defined as follows:
\[
\phi\ddfn\delta\mid \delta\leq\delta \mid \pci{\delta}{\delta}{\delta} \mid \wcn\phi \mid \phi\wedge\phi \mid \phi\vvee\phi \mid \existso x\phi \mid \forallo x\phi.
\]
Atoms of the form $\delta\leq\delta$ and $\pci{\delta}{\delta}{\delta}$ are called \textit{extended probabilistic inclusion} and \textit{extended probabilistic conditional independence} atoms, respectively. The fragment of $\FOPTinccond$ without extended probabilistic conditional independence atoms is denoted by $\FOPTinc$.

The semantics for $\FOPTinccond$ is defined as follows:
\begin{itemize}
\item $\mA\models_{\X}\delta$ iff $\mA\models_{s}\delta$ for all $s\in \supp(\X)$. 
\item $\mA\models_{\X}\delta_0\leq\delta_1$ iff $\sum_{s\in S_0}\X(s)\leq\sum_{s\in S_1}\X(s)$, where $S_i=\{s\in X\mid \mA\models_s\delta_i\}$ for $i=0,1$.
\item $\mA\models_{\X}\pci{\delta_0}{\delta_1}{\delta_2}$ iff 
\[
\sum_{s\in S_0\cap S_1}\X(s)\cdot\sum_{s\in S_0\cap S_2}\X(s)=\sum_{s\in S_0}\X(s)\cdot\sum_{s\in S_0\cap S_1\cap S_2}\X(s),
\]
where $S_i=\{s\in X\mid \mA\models_s\delta_i\}$ for $i=0,1,2$.
\item $\mA\models_{\X}\wcn\phi$ iff $\mA\not\models_{\X}\phi$ or $\supp(\X)=\varnothing$.
\item $\mA\models_{\X}\phi\wedge\psi$ iff $\mA\models_{\X}\phi$ and $\mA\models_{\X}\psi$.
\item $\mA\models_{\X}\phi\vvee\psi$ iff $\mA\models_{\X}\phi$ or $\mA\models_{\X}\psi$.
\item $\mA\models_{\X}\existso x\phi$ iff $\mA\models_{\X(a/x)}\phi$ for some $a\in A$. 
\item $\mA\models_{\X}\forallo x\phi$ iff $\mA\models_{\X(a/x)}\phi$ for all $a\in A$.
\end{itemize}
Note that  if $\X$ is an empty probabilistic team, then $\mA\models_{\X}\phi$ for any $\phi\in\FOPTinccond[\tau]$. The following proposition can also be verified using a simple induction:
\begin{prop}\label{distrprop}
Let $\X\colon X\to\mathbb{R}_{\geq 0}$ be a nonempty probabilistic team. Then for any formula $\phi\in\FOPTinccond[\tau]$ and any $\tau$-structure $\mA$
\[
\mA\models_{\distr(\X)}\phi\iff \mA\models_{\X}\phi.
\]
\end{prop}
Proposition \ref{distrprop} and its proof is similar to one from \cite{jelia19} which considers team-based logics with several different atoms, including marginal identity and probabilistic conditional independence (see also subsection \ref{marg&prop_cond_ind}). 

Next, we define a few notions that are needed to formulate the so-called \textit{locality} property. For a formula $\phi$, we denote by $\Var(\phi)$ the set of the free variables of $\phi$. Let $V$ be a set of variables. We write $s\restriction{V}$ for the restriction of the assignment $s$ to $V$. The \textit{restriction of a team} $X$ to $V$ is defined as $X\restriction{V}=\{s\restriction{V}\mid s\in X\}$. The \textit{restriction of a probabilistic team} $\X$ to $V$ is defined as $\X\restriction{V}\colon X\restriction{V}\to \mathbb{R}_{\geq 0}$ where
\[
(\X\restriction{V})(s)=\sum_{\substack{s'\restriction{V}=s,\\ s'\in X}}\X(s').
\]
\begin{prop}[\textbf{Locality}]\label{locality}
Let $\phi$ be any $\FOPTinccond[\tau]$-formula. Then for any set of variables $V$, any $\tau$-structure $\mA$, and any probabilistic team $\X\colon X\to\mathbb{R}_{\geq 0}$ such that $\Var(\phi)\subseteq V\subseteq D$, 
\[
\mA\models_{\X}\phi \iff \mA\models_{\X\restriction{V}}\phi.
\]
\end{prop}
\begin{proof}
By induction.  If $\phi=\delta$, the claim immediately holds since $\mA\models_{s}\delta\iff\mA\models_{s\restriction{V}}\delta$ for all $s\in X$. The cases $\phi=\theta_0\wedge\theta_1$ and $\phi=\theta_0\vvee\theta_1$ directly follow from the induction hypothesis.

For the cases $\phi=\delta_0\leq\delta_1$ and $\phi=\pci{\delta_0}{\delta_1}{\delta_2}$, we notice that 
\[
\sum_{s'\in S\restriction{V}}(\X\restriction{V})(s')=\sum_{s'\in S\restriction{V}}\left(\sum_{\substack{s\restriction{V}=s',\\ s\in X}}\X(s)\right)=\sum_{s\in S}\X(s),
\]
where $S=\{s\in X\mid \mA\models_s\delta\}$ and $S\restriction{V}=\{s'\in X\restriction{V}\mid \mA\models_{s'}\delta\}$ for any 
$\delta$. Then
\begin{align*}
\mA\models_{\X}\delta_0\leq\delta_1 \iff &\sum_{s\in S_0}\X(s)\leq \sum_{s\in S_1}\X(s), \text{ where }\ S_i=\{s\in X\mid \mA\models_s\delta_i\}\text{ for } i=0,1\\
\iff &\sum_{s'\in S_0\restriction{V}}(\X\restriction{V})(s')\leq \sum_{s'\in S_1\restriction{V}}(\X\restriction{V})(s'),\text{ where } S_i\restriction{V}=\{s'\in X\restriction{V}\mid \mA\models_{s'}\delta_i\} \text{ for } i=0,1\\
\iff &\mA\models_{\X\restriction{V}}\delta_0\leq\delta_1.
\end{align*}
The proof is similar for the case $\phi=\pci{\delta_0}{\delta_1}{\delta_2}$.

If $\phi=\wcn\theta_0$, then
\begin{align*}
\mA\models_{\X}\wcn\theta_0 \iff &\mA\not\models_{\X}\theta_0 \text{ or } \supp(\X)=\varnothing\\
\iff &\mA\not\models_{\X\restriction{V}}\theta_0 \text{ or } \supp(\X\restriction{V})=\varnothing\quad\text{(by the induction hypothesis)}\\
\iff &\mA\models_{\X\restriction{V}}\wcn\theta_0.
\end{align*}

If $\phi=\existso x\theta_0$, then
\begin{align*}
\mA\models_{\X}\existso x\theta_0 \iff &\mA\models_{\X(a/x)}\theta_0 \text{ for some } a\in A\\
\iff &\mA\models_{\X(a/x)\restriction{(V\cup\{x\})}}\theta_0 \text{ for some } a\in A\quad\text{(by the induction hypothesis)}\\
\iff &\mA\models_{(\X\restriction{V})(a/x)}\theta_0 \text{ for some } a\in A\quad(\text{since }\X(a/x)\restriction{(V\cup\{x\})}=(\X\restriction{V})(a/x))\\
\iff &\mA\models_{\X\restriction{V}}\existso x\theta_0.
\end{align*}
The proof is similar for the case $\phi=\forallo x\theta_0$.
\end{proof}
The next proposition shows that the quantifier-induced modifications of  probabilistic teams can also be viewed as substitution of quantified variables with suitable constants. We use this proposition in the proofs of Proposition \ref{varprop} and Theorem \ref{team2fo}. Let $\phi$ be a formula. We denote by $\phi_{(\bar{a}/\bar{x})}$ the formula obtained from $\phi$ by substituting the free occurrences of variables $\bar{x}$ with the constant symbols $\bar{a}$. When using the notation $\phi_{(\bar{a}/\bar{x})}$, we assume that the vocabulary of the model we are considering is complemented with the constant symbols $\bar{a}$.
\begin{prop}\label{lemmaprop}
Let $\phi$ be any $\FOPTinccond[\tau]$-formula. Then for any $\tau$-structure $\mA$, any probabilistic team $\X$, any tuple of variables $\bar{x}$, and any sequence $\bar{a}\in A^{|\bar{x}|}$
\[
\mA\models_{\X(\bar{a}/\bar{x})}\phi \iff \mA\models_{\X}\phi_{(\bar{a}/\bar{x})}.
\]
\end{prop}
\begin{proof}
If $\phi=\delta$, then
\begin{align*}
\mA\models_{\X(\bar{a}/\bar{x})}\delta \iff &\text{for all } s\colon D\cup\Var(\bar{x})\to A, \text{ if } s\in\supp(\X(\bar{a}/\bar{x})), \text{ then }\mA\models_{s}\delta\\
\iff &\text{for all } s\colon D\cup\Var(\bar{x})\to A, \text{ if } s\in\supp(\X(\bar{a}/\bar{x})), \text{ then }\mA\models_{s}\delta_{(\bar{a}/\bar{x})} \\ &(\text{if } s\in\supp(\X(\bar{a}/\bar{x})), \text{ then } s(\bar{x})=\bar{a})\\
\iff &\mA\models_{\X(\bar{a}/\bar{x})}\delta_{(\bar{a}/\bar{x})}\\
\iff &\mA\models_{\X}\delta_{(\bar{a}/\bar{x})}\quad (\text{by locality since }\X(\bar{a}/\bar{x})\restriction{(D\backslash\Var(\bar{x}))}=\X\restriction{(D\backslash\Var(\bar{x}))}).
\end{align*}
For the cases $\phi=\delta_0\leq\delta_1$ and $\phi=\pci{\delta_0}{\delta_1}{\delta_2}$, we notice that 
\[
\sum_{s\in S}\X(\bar{a}/\bar{x})(s)=\sum_{s'\in S'}\X(\bar{a}/\bar{x})(s'),
\]
where $S=\{s\colon D\cup\Var(\bar{x})\to A\mid\mA\models_{s}\delta\}$ and $S'=\{s'\colon D\cup\Var(\bar{x})\to A\mid\mA\models_{s'}\delta_{(\bar{a}/\bar{x})}\}$ for any $\delta$. For this, first note that if $s(\bar{x})\neq \bar{a}$, then $\X(\bar{a}/\bar{x})(s)=0$. Therefore, only those assignments $s$ for which $s(\bar{x})=\bar{a}$ may contribute to the sums. For those assignments $s$, clearly $\mA\models_{s}\delta \iff \mA\models_{s}\delta_{(\bar{a}/\bar{x})}$, and therefore $\sum_{s\in S}\X(\bar{a}/\bar{x})(s)=\sum_{s'\in S'}\X(\bar{a}/\bar{x})(s')$. With this, it is straightforward to check that the claim holds for the cases $\phi=\delta_0\leq\delta_1$ and $\phi=\pci{\delta_0}{\delta_1}{\delta_2}$.

If $\phi=\wcn\theta_0$, then 
\begin{align*}
\mA\models_{\X(\bar{a}/\bar{x})}\wcn\theta_0 \iff &\mA\not\models_{\X(\bar{a}/\bar{x})}\theta_0 \text{ or } \supp(\X(\bar{a}/\bar{x}))=\varnothing\\
\iff &\mA\not\models_{\X}{\theta_0}_{(\bar{a}/\bar{x})} \text{ or } \supp(\X)=\varnothing\quad\text{(by the induction hypothesis)}\\
\iff &\mA\models_{\X}\wcn{\theta_0}_{(\bar{a}/\bar{x})}.
\end{align*}

The proofs for the cases $\phi=\theta_0\wedge\theta_1$ and $\phi=\theta_0\vvee\theta_1$ directly follow from the induction hypothesis.

If $\phi=\existso y\theta_0$, then
\begin{align*}
\mA\models_{\X(\bar{a}/\bar{x})}\existso y\theta_0\iff &\mA\models_{\X(\bar{a}b/\bar{x}y)}\theta_0 \text{ for some }b\in A\\
\iff &\mA\models_{\X}{\theta_0}_{(\bar{a}b/\bar{x}y)} \text{ for some }b\in A\quad(\text{by the induction hypothesis})\\
\iff &\mA\models_{\X(b/y)}{\theta_0}_{(\bar{a}/\bar{x})} \text{ for some }b\in A\quad(\text{by the induction hypothesis})\\
\iff &\mA\models_{\X}\existso y{\theta_0}_{(\bar{a}/\bar{x})}.
\end{align*}
The proof is similar for the case $\phi=\forallo y\theta_0$. 
\end{proof}

The next proposition shows that we can rename quantified variables in the formulas. This is used in the proofs of Theorems \ref{constr_psi} and \ref{team2fo}, where we assume that certain variables have no bounded occurrences in the formulas. We introduce a notation that is analogous to $\phi_{(\bar{a}/\bar{x})}$: we write $\phi_{(\bar{y}/\bar{x})}$ for the formula where, instead of the constant symbols $\bar{a}$, we substitute $\bar{x}$ with the variables $\bar{y}$. 


\begin{prop}\label{varprop}
Let $\theta$ be any $\FOPTinccond[\tau]$-formula with free variables from $\{v_1,\dots ,v_k\}$. Suppose that $x$ does not appear in $\theta$. Then for any $\tau$-structure $\mA$, any probabilistic team $\X$ over $\{v_1,\dots ,v_k\}$, any $Q\in\{\existso,\forallo\}$, and any $w\in\{v_1,\dots ,v_k\}$
\[
\mA\models_{\X}Qw\theta\iff \mA\models_{\X}Qx\theta_{(x/w)}.
\]
\end{prop}
\begin{proof}
Define $\X_{x/w}\colon X_{x/w}\to A$ as the probabilistic team such that $X_{x/w}=\{s'\mid s\in X\}$ is the team over $(\{v_1,\dots ,v_k\}\backslash\{w\})\cup\{x\}$ where $s'(v_i)=s(v_i)$ when $v_i\neq w$, $s'(x)=s(w)$, and $\X_{x/w}(s')=\X(s)$. Thus the probabilistic team $\X_{x/w}$ is otherwise the same as the team $\X$ but the variable $w$ is replaced with $x$. Now we have
\begin{align*}
\mA\models_{\X}Q\bar{w}\theta\iff &\mA\models_{\X(a/w)}\theta \quad \text{for some/all } a\in A\\
\iff &\mA\models_{{\X(a/w)}_{x/w}}\theta_{(x/w)}\quad \text{for some/all } a\in A\\
\iff &\mA\models_{{\X}_{x/w}(a/x)}\theta_{(x/w)}\quad \text{for some/all } a\in A\\
\iff &\mA\models_{{\X}_{x/w}}\theta_{(x/w)(a/x)} \quad \text{for some/all } a\in A \quad(\text{by Prop. \ref{lemmaprop}})\\
\iff &\mA\models_{\X}\theta_{(x/w)(a/x)}\quad \text{for some/all } a\in A \quad\text{(by locality since } \\
&\X_{x/w}{\restriction{(\Var(\bar{v})\backslash\{w\})}}=\X{\restriction{(\Var(\bar{v})\backslash\{w\})}}\text{ )}\\
\iff &\mA\models_{\X(a/x)}\theta_{(x/w)}\quad  \text{for some/all } a\in A \quad(\text{by Prop. \ref{lemmaprop}})\\
\iff &\mA\models_{\X}Qx\theta_{(x/w)}.
\end{align*}
\end{proof}

\subsection{The logic \FOPTcondineq}

Next, we define a logic similar to $\FOPTinccond$. The difference is that, instead of the extended probabilistic inclusion and extended probabilistic conditional independence atoms, we have atoms of the form $(\delta_0|\delta_1)\leq(\delta_2|\delta_3)$, where $\delta_i$ is defined as in Equation \ref{delta}. We call these \textit{conditional probability inequality} atoms. The logic $\FOPTcondineq$ over a vocabulary $\tau$ is defined as follows:
\[
\phi\ddfn\delta\mid (\delta|\delta)\leq(\delta|\delta) \mid\wcn\phi \mid \phi\wedge\phi \mid \phi\vvee\phi \mid \existso x\phi \mid \forallo x\phi.
\]

The semantics for the atom $(\delta_0|\delta_1)\leq(\delta_2|\delta_3)$ is defined as follows:
\begin{align*}
\mA\models_{\X}(\delta_0|\delta_1)\leq(\delta_2|\delta_3)\iff \sum_{s\in S_0\cap S_1}\X(s)\cdot\sum_{s\in S_3}\X(s)\leq\sum_{s\in S_2\cap S_3}\X(s)\cdot\sum_{s\in S_1}\X(s)
\end{align*}
where $S_i=\{s\in X\mid \mA\models_s\delta_i\}$ for $i=0,1,2,3$.
Extended probabilistic inclusion and extended probabilistic conditional independence can be expressed in $\FOPTcondineq$. Suppose that $\delta_0,\delta_1,\delta_2$ are 
 formulas with free variables from $\bar{x}=(x_1,\dots,x_n)$. It is easy to check that
\[
\delta_0\leq\delta_1\equiv(\delta_0|x_1=x_1)\leq(\delta_1|x_1=x_1)
\]
and
\[
\pci{\delta_0}{\delta_1}{\delta_2}\equiv(\delta_1|\delta_0)\approx(\delta_1|\delta_0\wedge\delta_2),
\]
where $(\delta_1|\delta_0)\approx(\delta_1|\delta_0\wedge\delta_2)$ is an abbreviation for the formula $(\delta_1|\delta_0)\leq(\delta_1|\delta_0\wedge\delta_2)\wedge(\delta_1|\delta_0\wedge\delta_2)\leq(\delta_1|\delta_0)$.

Note that $\FOPTcondineq$ is local since the proof of Proposition \ref{locality} can easily be extended to cover atoms of the form $(\delta_0|\delta_1)\leq(\delta_2|\delta_3)$. Moreover, proofs for Propositions \ref{distrprop}, \ref{lemmaprop}, and \ref{varprop} can also be extended for $\FOPTcondineq$. 

\section{Comparison of logics in team semantics}

\subsection{$\FOPTinc$ as a generalization of $\FOT$ and $\FOT^{\downarrow}$}

The logic $\FOPTinc$ can be seen as a  generalization of $\FOT$ and $\FOT^{\downarrow}$ in the following sense:
\begin{prop} Let $\phi$ be any $\FOT[\tau]$-formula or $\FOT^{\downarrow}[\tau]$-formula. Then there exists an $\FOPTinc[\tau]$-formula  $\psi_{\phi}$ such that for any $\tau$-structure $\mA$, and any probabilistic team $\X$
\[
\mA\models_{\supp(\X)}\phi \iff \mA\models_{\X}\psi_{\phi}.
\]
\end{prop}
\begin{proof}
Notice that only inclusion atoms, i.e. atoms of the form $\bar{v}_0\subseteq\bar{v}_1$ need to be translated. For each formula $\phi$, we let $\psi_{\phi}$ be the same as $\phi$, except that each inclusion atom $\theta$ appearing in $\phi$ is substituted with the formula $\psi_{\theta}$ as described below. Provided that we can successfully translate each $\theta$, it is easy to check that the claim holds. If $\theta=\bar{v}_0\subseteq\bar{v}_1$, then we define $\psi_{\theta}:=\forallo \bar{x}(\neg\bar{v}_0=\bar{x}\vvee\wcn\neg\bar{v}_1=\bar{x})$. We show that the claim holds for $\theta$ and $\psi_{\theta}$.

If $\supp(\X)=\varnothing$, then both $\X$ and $\supp(\X)$ satisfy every formula. Thus, without loss of generality, we may assume that $\supp(\X)\neq\varnothing$. Now
\begin{align*}
\mA\models_{\supp(\X)}\bar{v}_0\subseteq\bar{v}_1 \iff &\text{ for all } s\in\supp(\X), \text{ there exists } s'\in \supp(\X) \text{ such that } s(\bar{v}_0)=s'(\bar{v}_1)\\
\iff &\text{ for all } \bar{a}\in A^{|\bar{v}_0|}, \text{ if there is }s\in \supp(\X) \text{ such that } s(\bar{v}_0)=\bar{a},\\
&\text{ then there exists } s'\in \supp(\X)\text{ such that }s'(\bar{v}_1)=\bar{a}\\
\iff &\text{ for all } \bar{a}\in A^{|\bar{v}_0|}, \text{ }  s(\bar{v}_0)\neq\bar{a} \text{ for all } s\in \supp(\X),\\
& \text{ or there exists } s'\in \supp(\X) \text{ such that } s'(\bar{v}_1)=\bar{a}\\
\iff &\text{ for all } \bar{a}\in A^{|\bar{v}_0|},\text{ } \mA\models_{\X}\neg\bar{v}_0=\bar{a} \text{ or } \mA\not\models_{\X}\neg\bar{v}_1=\bar{a}\\
\iff &\text{ for all } \bar{a}\in A^{|\bar{v}_0|},\text{ } \mA\models_{\X}\neg\bar{v}_0=\bar{a}\vvee\wcn\neg\bar{v}_1=\bar{a}\quad(\text{since }\supp(\X)\neq\varnothing)\\
\iff &\text{ for all } \bar{a}\in A^{|\bar{v}_0|},\text{ } \mA\models_{\X(\bar{a}/\bar{x})}\neg\bar{v}_0=\bar{x}\vvee\wcn\neg\bar{v}_1=\bar{x}\quad( \text{by Prop. \ref{lemmaprop}})\\
\iff & \mA\models_{\X}\forallo \bar{x}(\neg\bar{v}_0=\bar{x}\vvee\wcn\neg\bar{v}_1=\bar{x}).
\end{align*}
\end{proof}

\subsection{Expressivity of marginal identity and probabilistic conditional independence atoms in $\FOPTinccond$}\label{marg&prop_cond_ind}

The logics in probabilistic team semantics often include the \textit{marginal identity atom} $\bar{v}_0\approx\bar{v}_1$ and the \textit{probabilistic conditional independence atom} $\pci{\bar{v}_0}{\bar{v}_1}{\bar{v}_2}$ where $\bar{v}_0,\bar{v}_1$ and $\bar{v}_2$ are tuples of variables, instead of formulas. (See e.g. \cite{jelia19}.) In the case of the marginal identity atom, we additionally require that $|\bar{v}_0|=|\bar{v}_1|$. We first give semantics for these atoms, and then show that the atoms of the form $\delta_0\leq\delta_1$ and $\pci{\delta_0}{\delta_1}{\delta_2}$ extend these in the sense that when the weak universal quantifier $\forallo$ is available, $\bar{v}_0\approx\bar{v}_1$ and $\pci{\bar{v}_0}{\bar{v}_1}{\bar{v}_2}$ are also expressible. 

Let $\bar{x}$ be a tuple of variables and $\bar{a}\in A^{|\bar{x}|}$, and define 
\[
|\X_{\bar{x}=\bar{a}}|:=\sum_{\substack{s\in X,\\s(\bar{x})=\bar{a}}} \X(s).
\]
The semantics for the marginal identity atom and the probabilistic conditional independence atom is defined as follows:
\begin{itemize}
\item $\mA\models_{\X}\bar{v}_0\approx\bar{v}_1$ iff $|\X_{\bar{v}_0=\bar{a}}|=|\X_{\bar{v}_1=\bar{a}}|$ for all $\bar{a}\in A^{|\bar{v}_0|}$.
\item $\mA\models_{\X}\pci{\bar{v}_0}{\bar{v}_1}{\bar{v}_2}$ iff \[
|\X_{\bar{v}_0\bar{v}_1=s(\bar{v}_0\bar{v}_1)}|\cdot|\X_{\bar{v}_0\bar{v}_2=s(\bar{v}_0\bar{v}_2)}=|\X_{\bar{v}_0=s(\bar{v}_0)}|\cdot|\X_{\bar{v}_0\bar{v}_1\bar{v}_2=s(\bar{v}_0\bar{v}_1\bar{v}_2)}|
\] for all $s\colon\Var(\bar{v}_0\bar{v}_1\bar{v}_2)\to A$.
\end{itemize}
For probabilistic conditional independence, the equivalent formula of $\FOPTinccond$ is straightforward to obtain:
 \[
\pci{\bar{v}_0}{\bar{v}_1}{\bar{v}_2}\equiv\forallo \bar{x}\bar{y}\bar{z}(\pci{\bar{v}_0=\bar{x}}{\bar{v}_1=\bar{y}}{\bar{v}_2=\bar{z}}).
\]
For the marginal identity atom, it feels natural to first define a new kind of formula $\delta_0\approx\delta_1:=\delta_0\leq\delta_1\wedge\delta_1\leq\delta_0$,
and use that to obtain that
\[
\bar{v}_0\approx\bar{v}_1\equiv\forallo\bar{x}(\bar{v}_0=\bar{x}\approx\bar{v}_1=\bar{x}).
\]
However, there is also a shorter formula for the marginal identity atom:
\[
\bar{v}_0\approx\bar{v}_1\equiv\forallo\bar{x}(\bar{v}_0=\bar{x}\leq\bar{v}_1=\bar{x}).
\]
To see that this formula suffices, note that since $A_0$ is finite,
\[
|\X_{\bar{v}_0=\bar{a}}|\leq|\X_{\bar{v}_1=\bar{a}}| \text{ for all } \bar{a}\in A^{|\bar{v}_0|}
\]
implies that 
\[
|\X_{\bar{v}_0=\bar{a}}|=|\X_{\bar{v}_1=\bar{a}}| \text{ for all } \bar{a}\in A^{|\bar{v}_0|}.
\]
Because of this, marginal identity atoms were originally (in \cite{DurandHKMV18}) called \textit{probabilistic inclusion atoms} and denoted by $\bar{v}_0\leq\bar{v}_1$. Instead of defining the formula $\delta_0\approx\delta_1$ as we have done above, we could also treat it as a new kind of atomic formula. Then the atoms of the form $\delta_0\leq\delta_1$ and $\delta_0\approx\delta_1$ can be seen as extended probabilistic inclusion and extended marginal identity atoms, respectively. However, even though the truth definitions for $\bar{v}_0\leq\bar{v}_1$ and $\bar{v}_0\approx\bar{v}_1$ are equivalent, this is not the case for $\delta_0\leq\delta_1$ and $\delta_0\approx\delta_1$.

\section{Translation from $\FOPTcondineq$ to real arithmetic}\label{team2real}

In this section, we show that the satisfiability and validity problems for $\FOPTcondineq$ are \textsc{RE}-complete and co-\textsc{RE}-complete, respectively. The main ingredient of the proof is constructing a translation from $\FOPTcondineq$ to real arithmetic.

We say that a $\tau$-formula $\phi\in\FOPTcondineq$ is \textit{satisfiable in a} $\tau$-\textit{structure} $\mA$ if there exists a nonempty probabilistic team $\X$ of $\mA$ such that $\mA\models_{\X}\phi$. Analogously, $\phi$ is \textit{valid in} $\mA$ if $\mA\models_{\X}\phi$ for all probabilistic teams $\X$ of $\mA$ over $\Var(\phi)$. A $\tau$-formula $\phi\in\FOPTcondineq$ is \textit{satisfiable} if there exists a $\tau$-structure $\mA$ such that $\phi$ is satisfiable in $\mA$. A $\tau$-formula $\phi\in\FOPTcondineq$ is \textit{valid} if $\phi$ is valid in $\mA$ for all a $\tau$-structures $\mA$.

\begin{theorem}\label{constr_psi}
Let $\tau$ be a finite relational vocabulary, and $\mA$ a finite $\tau$-structure. 
\begin{itemize}
\item[(i)] For each $\tau$-formula $\phi\in\FOPTinc$ there exists a first-order sentence $\psi$ over vocabulary $\{+,\leq,0\}$ such that $\phi$ is satisfiable in $\mA$ iff $(\mathbb{R},+,\leq,0)\models \psi$.
\item[(ii)] For each $\tau$-formula $\phi\in\FOPTinccond$ there exists a first-order sentence $\psi$ over vocabulary $\{+,\times,\leq,0,1\}$ such that $\phi$ is satisfiable in $\mA$ iff $(\mathbb{R},+,\times,\leq,0,1)\models \psi$.
\item[(iii)] For each $\tau$-formula $\phi\in\FOPTcondineq$ there exists a first-order sentence $\psi$ over vocabulary $\{+,\times,\leq,0,1\}$ such that $\phi$ is satisfiable in $\mA$ iff $(\mathbb{R},+,\times,\leq,0,1)\models \psi$.
\end{itemize} 
\end{theorem}
\begin{proof} Without loss of generality, we may assume that $A=\{1,\dots,n\}$. Let $\bar{v}=(v_1,\dots ,v_m)$ be a tuple that consists of the first-order variables that appear free in $\phi$. Since $\FOPTinc$, $\FOPTinccond$, and $\FOPTcondineq$ are local, it suffices to consider teams over $\{v_1,\dots ,v_m\}$. Moreover, by Proposition \ref{varprop}, it suffices to only consider formulas $\phi(\bar{v})$ in which there are no bound occurrences of the variables $\bar{v}$. For the tuple $\bar{v}$, we will need a fresh first-order variable $s_{\bar{v}=\bar{i}}$ for each $\bar{i}\in A^m$. Each variable $s_{\bar{v}=\bar{i}}$ will correspond to the weight of the assignment that interprets variables $\bar{v}$ as elements $\bar{i}$. By $\bar{s}$, we denote the tuple $(s_{\bar{v}=\bar{1}},\dots  ,s_{\bar{v}=\bar{n}})$ that contains all these variables. Now we define

\[
\psi:=\exists s_{\bar{v}=\bar{1}}\dots  s_{\bar{v}=\bar{n}}\left(\bigwedge_{\bar{i}} 0\leq s_{\bar{v}=\bar{i}}\wedge \neg 0=\sum_{\bar{i}} s_{\bar{v}=\bar{i}}\wedge\phi^*(\bar{s})\right) ,
\]
where $\phi^*(\bar{s})$ is defined inductively as follows:
\begin{itemize}
\item[-] If $\phi(\bar{v})=\delta$, then $\phi^*(\bar{s}):=\bigwedge_{s\in S}s=0$, where 
\[
S=\{s\in\{s_{\bar{v}=\bar{1}},\dots , s_{\bar{v}=\bar{n}}\}\mid\mA\not\models_s\delta\}.
\]
\item[-] If $\phi(\bar{v})=\delta_0(\bar{v})\leq\delta_1(\bar{v})$, then 
\[
\phi^*(\bar{s}):=\sum_{s\in S_0} s\leq\sum_{s\in S_1} s,
\]
where $S_i=\{s\mid \mA\models_s\delta_i\}$ for $i=0,1$.
\item[-] If $\phi(\bar{v})=\pci{\delta_0(\bar{v})}{\delta_1(\bar{v})}{\delta_2(\bar{v})}$, then 
\[
\phi^*(\bar{s}):=\sum_{s\in S_0\cap S_1} s\times \sum_{s\in S_0\cap S_2} s = \sum_{s\in S_0} s \times \sum_{s\in S_0\cap S_1\cap S_2} s,
\]
where $S_i=\{s\mid \mA\models_s\delta_i\}$ for $i=0,1,2$.
\item[-] If $\phi(\bar{v})=(\delta_0(\bar{v})\mid\delta_1(\bar{v}))\leq(\delta_2(\bar{v})\mid\delta_3(\bar{v}))$, then 
\[
\phi^*(\bar{s}):=\sum_{s\in S_0\cap S_1} s\times \sum_{s\in S_3} s \leq \sum_{s\in S_2\cap S_3} s \times \sum_{s\in S_1} s,
\]
where $S_i=\{s\mid \mA\models_s\delta_i\}$ for $i=0,1,2,3$.
\item[-] If $\phi(\bar{v})=\wcn\theta_0(\bar{v})$, then $\phi^*(\bar{s}):=\neg\theta_0^*(\bar{s})$.
\item[-] If $\phi(\bar{v})=\theta_0(\bar{v})\wedge\theta_1(\bar{v})$, then $\phi^*(\bar{s}):=\theta_0^*(\bar{s})\wedge\theta_1^*(\bar{s})$.
\item[-] If $\phi(\bar{v})=\theta_0(\bar{v})\vvee\theta_1(\bar{v})$, then $\phi^*(\bar{s}):=\theta_0^*(\bar{s})\vee\theta_1^*(\bar{s})$.
\item[-] If $\phi(\bar{v})=\exists x\theta_0(\bar{v},x)$, then 
\begin{align*}
\phi^*(\bar{s}):=\exists t_{\bar{v}x=\bar{1}1} \dots t_{\bar{v}x=\bar{n}n}\biggl(\bigvee_{j}\bigwedge_{\bar{i}}( t_{\bar{v}x=\bar{i}j}=s_{\bar{v}=\bar{i}}\wedge\bigwedge_{k\neq j}t_{\bar{v}x=\bar{i}k}=0)\wedge\theta_0^*(\bar{t})\biggr).
\end{align*}
\item[-] If $\phi(\bar{v})=\forall x\theta_0(\bar{v},x)$, then 
\begin{align*}
\phi^*(\bar{s}):=\bigwedge_{j}\biggl(\exists t_{\bar{v}x=\bar{1}1} \dots t_{\bar{v}x=\bar{n}n}\biggl(\bigwedge_{\bar{i}}( t_{\bar{v}x=\bar{i}j}=s_{\bar{v}=\bar{i}}\wedge\bigwedge_{k\neq j}t_{\bar{v}x=\bar{i}k}=0)\wedge\theta_0^*(\bar{t})\biggr)\biggr).
\end{align*}
\end{itemize}
\end{proof}
Let $\mathcal{L}$ be a logic. We denote by $\textsc{SAT}(\mathcal{L})$ and $\textsc{VAL}(\mathcal{L})$ the satisfiability and the validity problems for $\mathcal{L}$, respectively. 
\begin{theorem}
The satisfiability problem for $\FOPTcondineq$ is \textsc{RE}-complete.
\end{theorem} 
\begin{proof}
\underline{Inclusion}: Suppose that $\phi\in\FOPTcondineq[\tau]$ is satisfiable. Let $\mA$ be any finite $\tau$-structure. By Theorem \ref{constr_psi}, we can construct a sentence $\psi_{\mA,\phi}$ such that $\phi$ is satisfiable in $\mA$ iff $(\mathbb{R},+,\times,\leq,0,1)\models\psi_{\mA,\phi}$. Note that the sentence $\psi_{\mA,\phi}$ is computable since $\mA\models_s\delta$ is decidable when structure $\mA$, assignment $s$, and formula $\delta$ are given. Since truth in real arithmetic is decidable, given a structure $\mA$, we can also decide whether $\phi$ is satisfiable in $\mA$. Thus we can  verify that $\phi$ is satisfiable by going through all finite $\tau$-structures until we find a structure $\mA$ such that $\phi$ is satisfiable in $\mA$. 

\underline{Hardness}: Notice that every first-order sentence is also expressible in $\FOPTcondineq$, and therefore $\textsc{SAT}(\logicFont{FO})$ (in the finite) is reducible to $\textsc{SAT}(\FOPTcondineq)$. By Trahtenbrot's Theorem, the halting problem is reducible to $\textsc{SAT}(\logicFont{FO})$. Since the halting problem is \textsc{RE}-complete, $\textsc{SAT}(\FOPTcondineq)$ is \textsc{RE}-hard.
\end{proof}
\begin{theorem}\label{val}
The validity problem for $\FOPTcondineq$ is \emph{co-}\textsc{RE}-complete.
\end{theorem} 
\begin{proof}
Denote by $\textsc{VAL}(\FOPTcondineq)$ the set of $\tau$-formulas that are valid, and by $\overline{\textsc{SAT}}(\FOPTcondineq)$ the set of $\tau$-formulas that are not satisfiable. Then 
\begin{align*}
\phi\in\textsc{VAL}(\FOPTcondineq)\iff & \mA\models_{\X}\phi \text{ for all } \mA \text{ and } \X\\
\iff &\mA\not\models_{\X}\wcn\phi \text{ for all } \mA \text{ and } \X\\
\iff & \wcn\phi\in\overline{\textsc{SAT}}(\FOPTcondineq)
\end{align*}
and
\begin{align*}
\phi\in\overline{\textsc{SAT}}(\FOPTcondineq)
\iff & \mA\not\models_{\X}\phi \text{ for all } \mA \text{ and } \X\\
\iff &\mA\models_{\X}\wcn\phi \text{ for all } \mA \text{ and } \X\\
\iff & \wcn\phi\in\textsc{VAL}(\FOPTcondineq).
\end{align*}
Thus $\textsc{VAL}(\FOPTcondineq)$ is reducible to $\overline{\textsc{SAT}}(\FOPTcondineq)$, and vice versa. Since $\textsc{SAT}(\FOPTcondineq)$ is \textsc{RE}-complete, $\overline{\textsc{SAT}}(\FOPTcondineq)$ is co-\textsc{RE}-complete, and therefore $\textsc{VAL}(\FOPTcondineq)$ is also co-\textsc{RE}-complete.
\end{proof}

\section{Counterparts of logics in probabilistic team semantics over metafinite structures}\label{FOrxsum}

In this section, we define two-sorted structures, and the logic $\FOrxsum$. We also show that there is no translation from $\FOrxsum$ to $\FOPTinccond$, and define a fragment $\FOrsum$ which is equi-expressive with the logic $\FOPTinc$.
\begin{defn}[A two-sorted structure]
Let $\tau_0$, $\tau_1$, and $\sigma$ be vocabularies such that $\sigma$ is functional, and $\tau_0\cap\sigma=\tau_1\cap\sigma=\varnothing$. A two-sorted structure of vocabulary $\tau_0\cup\tau_1\cup\sigma$ is a tuple $\mA=(\mA_0,\mA_1,F)$ where $\mA_i$ is a $\tau_i$-structure of domain $A_i$ for $i=0,1$, and $F$ is a set that contains functions $f^{\mA}\colon A_0^{\ar(f)}\to A_1$ for each function symbol $f\in\sigma$.
\end{defn}
In this paper, we always assume that the structure $\mA_0$ is finite, 
and both $\sigma$ and $F$ are finite. For simplicity, we also assume that $\tau_0$
only contains relation and constant symbols. Note that $\mA_1$ is not assumed to be finite, on the contrary, we consider metafinite structures where $A_1=\mathbb{R}_{\geq 0}$ or $A_1=\mathbb{R}$.

We let $\{=\}\subseteq\tau_0$, $\tau_1=\{\leq\}$, and $\sigma=\{f\}$. We consider structures $\mA=(\mA_0,\mA_1,F)$ where $\mA_0$ is a finite $\tau_0$-structure, $\mA_1=(\mathbb{R}_{\geq 0},\leq)$, and $F=\{f^{\mA}\}$ for some $f^{\mA}\colon A_0\to\mathbb{R}_{\geq 0}$. We call these structure $\mathbb{R}_{\geq 0}$-structures. We define a logic $\FOrxsum$ on $\mathbb{R}_{\geq 0}$-structures. First-order $\tau_0$-terms and atomic formulas are defined in the usual way. Let $\lambda$ be a first-order atomic $\tau_0$-formula, and define
\[
\gamma\ddfn\lambda\mid \neg\gamma \mid \gamma\wedge\gamma \mid \gamma\lor\gamma.
\]
Then, in addition to the usual $\tau_0$-terms, we have \textit{numerical} $\tau_0\cup\sigma$-terms $i$ which are defined as follows:
\[
i\ddfn f(\bar{y})\mid i\times i\mid\SUM_{\bar{x}}(i,\gamma),
\]
where $\bar{x}$ and $\bar{y}$ are tuples of variables and $|\bar{y}|=\ar(f)$. If $|\bar{x}|=0$, we denote $\SUM_{\bar{x}}(i,\gamma)=\SUM_{\varnothing}(i,\gamma)$.
The logic $\FOrxsum$ over a vocabulary $\tau_0\cup\tau_1\cup\sigma$ is then defined as follows:
\[
\phi\ddfn\lambda \mid i\leq i \mid \neg\phi  \mid \phi\wedge\phi \mid \phi\lor\phi \mid \exists x\phi \mid \forall x\phi,
\]
where $x$ is a first-order variable.

We now define the semantics for $\FOrxsum$. Let $\mA$ be an $\mathbb{R}_{\geq 0}$-structure of a vocabulary $\tau_0\cup\tau_1\cup\sigma$. The interpretations of $\tau_0$-terms are defined in the usual way. Note that first-order terms only range over $A_0$; they cannot take values from ${\mathbb{R}_{\geq 0}}$. For the numerical terms we define interpretations $[f(\bar{x})]_s^\mathcal{A}:=f^\mathcal{A}(s(\bar{x}))$,
\[
[i\times j]_s^\mathcal{A}:=[i]_s^\mathcal{A}\cdot [j]_s^\mathcal{A},
\]
and
\[
[\SUM_{\bar{x}}(i,\gamma)]_s^\mathcal{A}:=\sum_{\bar{a}\in B}[i]_{s(\bar{a}/\bar{x})}^\mathcal{A},
\]
where $B=\{\bar{a}\in A_0^{|\bar{x}|}\mid \mA_0\models_s\gamma(\bar{a}/\bar{x})\}$. 
The semantics for $\leq$ is defined in the obvious way, i.e.
\[
\mA\models_s i\leq j \iff [i]_s^\mathcal{A}\leq [j]_s^\mathcal{A}.
\]
For atomic $\tau_0$-formulas and connectives $\neg$, $\wedge$, $\lor$, $\exists x$, and $\forall x$, we define semantics as in first-order logic.

In Theorem \ref{team2fo}, we present a translation from $\FOPTcondineq$ to $\FOrxsum$.
However, Theorem \ref{notransl} shows that there is no full translation from $\FOrxsum$ to $\FOPTcondineq$.
On the other hand, there is a fragment of $\FOrxsum$ which is equi-expressive with the logic $\FOPTinc$ on $\mathbb{R}_{\geq 0}$-structures. (See Section \ref{equi}.) The fragment is denoted by $\FOrsum$, and defined as 
\[
\phi\ddfn\lambda \mid \neg\phi \mid \SUM_{\bar{x}}(f(\bar{y}),\gamma)\leq \SUM_{\bar{x}}(f(\bar{y}),\gamma) \mid \phi\wedge\phi \mid \phi\lor\phi \mid \exists x\phi \mid \forall x\phi
\]
where $\lambda$ and $\gamma$ are defined as before, and $\bar{x}$ and $\bar{y}$ are tuples of distinct variables such that $\Var(\bar{x})\subseteq\Var(\bar{y})$ and $|\bar{y}|=\ar(f)$. Note that despite the restricted syntax of the fragment, we can still refer to $f^{\mA}(s(\bar{y}))$ (and also the constant 0). For this, we notice that the set $A_0^{|\varnothing|}=A_0^0$ is the singleton containing only the empty tuple, and therefore
\[
[\SUM_{\varnothing}(f(\bar{y}),\gamma)]_s^\mathcal{A}=
\begin{cases}
f^{\mA}(s(\bar{y})), &\text{ when } \mA\models_s\gamma\\
0, &\text{ when } \mA\not\models_s\gamma.
\end{cases}
\]
Additionally, we define a useful abbreviation
\begin{align*}
i=j := &i\leq j\wedge j\leq i,
\end{align*}
and write  $f(\bar{u})=0$ for the formula 
\begin{align*}
&\SUM_{\varnothing}(f(\bar{u}),u_1=u_1)=\SUM_{\varnothing}(f(\bar{u}),\neg u_1=u_1),
\end{align*}
where $\bar{u}=(u_1,\dots ,u_k)$. Note that $[\SUM_{\varnothing}(f(\bar{u}),u_1=u_1)]_s^{\mA}=f^{\mA}(s(\bar{u}))$ and $[\SUM_{\varnothing}(f(\bar{u}),\neg u_1=u_1)]_s^{\mA}=0$, and thus
\[
\mA\models_sf(\bar{u})=0\iff f^{\mA}(s(\bar{u}))=0
\]
as one would expect.

\section{Translations and the equi-expressivity result}

\subsection{Translation from $\FOPTcondineq$ to $\FOrxsum$}\label{translation}

\begin{theorem}\label{team2fo}
Let $\phi(v_1,\dots ,v_k)$ be any $\FOPTcondineq[\tau_0]$-formula and $f$ a $k$-ary function symbol. Then there exists an $\FOrxsum[\tau_0\cup\{\leq\}\cup\{f\}]$-sentence $\psi_{\phi}(f)$ such that for any $\mathbb{R}_{\geq 0}$-structure $\mA=(\mA_0,\mA_1,\{f_{\X}\})$ and any probabilistic team $\X$ over $\{v_1,\dots ,v_k\}$
\[
\mA_0\models_{\X}\phi(\bar{v})\iff \mA\models\psi_{\phi}(f),
\]
where $f_{\X}\colon A_0^k\to\mathbb{R}_{\geq 0}$ is a function such that $f_{\X}(s(\bar{v}))=\X(s)$ for all $s\in X$.
\end{theorem}

\begin{proof} We show by induction that for any subformula $\theta(\bar{v},\bar{x})$ of $\phi(\bar{v})$, there exists an $\FOrxsum[\tau_0\cup\{\leq\}\cup\{f\}]$-formula $\psi_{\theta}(f,\bar{x})$ such that for any $\mathbb{R}_{\geq 0}$-structure $\mA=(\mA_0,\mA_1,\{f_{\X}\})$, any probabilistic team $\X$ over $\{v_1,\dots ,v_k\}$, and any sequence $\bar{a}\in A_0^{|\bar{x}|}$
\[
\mA_0\models_{\X(\bar{a}/\bar{x})}\theta(\bar{v},\bar{x})\iff \mA\models\psi_{\theta}(f,\bar{x})(\bar{a}/\bar{x}),
\]
where $f_{\X}\colon A_0^k\to\mathbb{R}_{\geq 0}$ is a function defined as above. Note that by Proposition \ref{varprop}, it suffices to only consider formulas $\phi(\bar{v})$ in which there are no bound occurrences of the variables $\bar{v}$.

\begin{itemize}
\item[(1)] Suppose that $\theta(\bar{v},\bar{x})=\delta(\bar{v},\bar{x})$. Then let $\psi_{\theta}(f,\bar{x}):=\forall \bar{u}(f(\bar{u})=0\lor\delta(\bar{u}/\bar{v},\bar{x}))$.

Now 
\begin{align*}
\mA_0\models_{\X(\bar{a}/\bar{x})}\delta(\bar{v},\bar{x})
\iff &\mA_0\models_{\X}\delta(\bar{v},\bar{x})_{(\bar{a}/\bar{x})}\quad (\text{by Prop. \ref{lemmaprop}})\\
\iff &\text{ for all } s\in X, \text{ if } s\in\supp(\X), \text{ then }\mA_0\models_{s}\delta(\bar{v},\bar{x})_{(\bar{a}/\bar{x})}\\
\iff &\text{ for all } \bar{b}\in A_0^k,\text{ }f_{\X}(\bar{b})=0\text{ or }\mA_0\models\delta(\bar{b}/\bar{v},\bar{x})(\bar{a}/\bar{x})\\
\iff &\mA\models\forall \bar{u}(f(\bar{u})=0\lor\delta(\bar{u}/\bar{v},\bar{x}))(\bar{a}/\bar{x}).
\end{align*}

\item[(2)] Suppose that $\theta(\bar{v},\bar{x})=(\delta_0|\delta_1)\leq(\delta_2|\delta_3)$. Then let 
\begin{align*}
\psi_{\theta}(f,\bar{x}):=&\SUM_{\bar{u}}(f(\bar{u}),(\delta_0\wedge\delta_1)(\bar{u}/\bar{v},\bar{x}))\times\SUM_{\bar{u}}(f(\bar{u}),\delta_3(\bar{u}/\bar{v},\bar{x}))\leq\\
&\SUM_{\bar{u}}(f(\bar{u}),(\delta_2\wedge\delta_3)(\bar{u}/\bar{v},\bar{x}))\times\SUM_{\bar{u}}(f(\bar{u}),\delta_1(\bar{u}/\bar{v},\bar{x})).
\end{align*}

Now
\begin{align*}
\mA_0\models_{\X(\bar{a}/\bar{x})}(\delta_0|\delta_1)\leq(\delta_2|\delta_3) \iff &\mA_0\models_{\X}((\delta_0|\delta_1)\leq(\delta_2|\delta_3))_{(\bar{a}/\bar{x})}\quad(\text{by Prop. \ref{lemmaprop}})\\
\iff &\sum_{s\in S_0\cap S_1}\X(s)\cdot\sum_{s\in S_3}\X(s)\leq\sum_{s\in S_2\cap S_3}\X(s)\cdot\sum_{s\in S_1}\X(s),\\
&\text{ where } S_i=\{s\in X\mid \mA_0\models_s{\delta_i}_{(\bar{a}/\bar{x})}\}\text{ for } i=0,1,2,3\\
\iff &\sum_{\bar{b}\in B_0\cap B_1}f_{\X}(\bar{b})\cdot\sum_{\bar{b}\in B_3}f_{\X}(\bar{b})\leq\sum_{\bar{b}\in B_2\cap B_3}f_{\X}(\bar{b})\cdot\sum_{\bar{b}\in B_1}f_{\X}(\bar{b}),\\
&\text{ where } B_i=\{\bar{b}\in A_0^k\mid \mA_0\models\delta_i(\bar{b}/\bar{v},\bar{a}/\bar{x})\}\text{ for } i=0,1,2,3\\
\iff &\mA\models\psi_{\theta}(f,\bar{x})(\bar{a}/\bar{x}).
\end{align*}

\item[(3)] Suppose that $\theta(\bar{v},\bar{x})=\wcn\theta_0(\bar{v},\bar{x})$. Then let 
\[
\psi_{\theta}(f,\bar{x}):=\neg\psi_{\theta_0}(f,\bar{x})\lor\forall\bar{u}f(\bar{u})=0.
\]

Now
\begin{align*}
\mA_0\models_{\X(\bar{a}/\bar{x})}\wcn\theta_0(\bar{v},\bar{x})
\iff &\mA_0\models_{\X}\wcn\theta_0(\bar{v},\bar{x})_{(\bar{a}/\bar{x})}\quad(\text{by Prop. \ref{lemmaprop}})\\
\iff &\mA_0\not\models_{\X}\theta_0(\bar{v},\bar{x})_{(\bar{a}/\bar{x})} \text{ or } \supp(\X)=\varnothing\\
\iff &\mA_0\not\models_{\X(\bar{a}/\bar{x})}\theta_0(\bar{v},\bar{x}) \text{ or } f_{\X}(\bar{b})=0 \text{ for all } \bar{b}\in A_0^k\quad(\text{by Prop. \ref{lemmaprop}})\\
\iff &\mA\not\models\psi_{\theta_0}(f,\bar{x})(\bar{a}/\bar{x}) \text{ or } f_{\X}(\bar{b})=0\text{ for all } \bar{b}\in A_0^k \quad\text{(by the induction hypothesis)}\\
\iff &\mA\models\neg\psi_{\theta_0}(f,\bar{x})(\bar{a}/\bar{x}) \text{ or } \mA\models\forall\bar{u}f(\bar{u})=0\\
\iff &\mA\models(\neg\psi_{\theta_0}(f,\bar{x})\lor\forall\bar{u}f(\bar{u})=0)(\bar{a}/\bar{x}).
\end{align*}

\item[(4)] Suppose that $\theta(\bar{v},\bar{x})=\theta_0(\bar{v},\bar{x})\wedge\theta_1(\bar{v},\bar{x})$, where $\theta(\bar{v},\bar{x})$ does not belong to the item (1). Then let $\psi_{\theta}(f,\bar{x}):=\psi_{\theta_0}(f,\bar{x})\wedge\psi_{\theta_1}(f,\bar{x})$. The claim directly follows from the induction hypothesis.

\item[(5)] Suppose that $\theta(\bar{v},\bar{x})=\theta_0(\bar{v},\bar{x})\vvee\theta_1(\bar{v},\bar{x})$. Then let $\psi_{\theta}(f,\bar{x}):=\psi_{\theta_0}(f,\bar{x})\lor\psi_{\theta_1}(f,\bar{x})$. The claim directly follows from the induction hypothesis.

\item[(6)] Suppose that $\theta(\bar{v},\bar{x})=\existso y\theta_0(\bar{v},\bar{x}y)$. Then let $\psi_{\theta}(f,\bar{x}):=\exists y\psi_{\theta_0}(f,\bar{x}y)$.
Now
\begin{align*}
\mA_0\models_{\X(\bar{a}/\bar{x})}\existso y\theta_0(\bar{v},\bar{x}y)
\iff &\mA_0\models_{\X(\bar{a}b/\bar{x}y)}\theta_0(\bar{v},\bar{x}y) \text{ for some } b\in A_0\\
\iff &\mA\models\psi_{\theta_0}(f,\bar{x}y)(\bar{a}b/\bar{x}y) \text{ for some } b\in A_0\quad\text{(by the induction hypothesis)}\\
\iff &\mA\models\exists y\psi_{\theta_0}(f,\bar{x}y)(\bar{a}/\bar{x}).
\end{align*}

\item[(7)] Suppose that $\theta(\bar{v},\bar{x})=\forallo y\theta_0(\bar{v},\bar{x}y)$. Then let $\psi_{\theta}(f,\bar{x}):=\forall x\psi_{\theta_0}(f,\bar{x}y)$. This is similar to case (6).
\end{itemize}

\end{proof}

The next theorem shows that the converse does not hold in full generality. We will show that the scaling property of $\FOPTcondineq$, i.e. Proposition \ref{distrprop}, fails for $\FOrxsum$.
\begin{theorem}\label{notransl}
Let $f$ be a $k$-ary function symbol. There exists a sentence $\psi\in\FOrxsum[\tau_0\cup\{\leq\}\cup\{f\}]$ for which there is no formula $\phi_{\psi}(v_1,\dots ,v_k)\in\FOPTcondineq[\tau_0]$ such that for any $\mathbb{R}_{\geq 0}$-structure $\mA=(\mA_0,\mA_1,\{f_{\X}\})$ and any nonempty probabilistic team $\X$ 
over $\{v_1,\dots ,v_k\}$ 
\[
\mA_0\models_{\X}\phi_{\psi}\iff \mA\models\psi,
\]
where $f_{\X}\colon A_0^k\to\mathbb{R}_{\geq 0}$ is a function such that $f_{\X}(s(\bar{v}))=\X(s)$ for all $s\in X$.
\end{theorem}
\begin{proof}
Let $x,y_1,\cdots ,y_k$ be variables such that $k=\ar(f)$, $\bar{y}=(y_1,\cdots ,y_k)$, and $x\notin\Var(\bar{y})$. Define 
\[
i_0:=\SUM_{\bar{y}}(f(\bar{y}),y_1=y_1),
\]
and 
\[
i_1:=\SUM_x(i_0,x=x).
\]
Let $\psi:=i_0\times i_0\leq i_1$. We show that $\psi$ is as wanted. For a contradiction, suppose that there is an equivalent formula $\phi_{\psi}$. We notice that 
\begin{align*}
[i_0\times i_0]^{\mA}_s&=[\SUM_{\bar{y}}(f(\bar{y}),y_1=y_1)\times\SUM_{\bar{y}}(f(\bar{y}),y_1=y_1)]^{\mA}_s\\
&=[\SUM_{\bar{y}}(f(\bar{y}),y_1=y_1)]^{\mA}_s\cdot[\SUM_{\bar{y}}(f(\bar{y}),y_1=y_1)]^{\mA}_s\\
&=\sum_{\bar{b}\in A_0^k}f_{\X}(\bar{b})\cdot \sum_{\bar{b}\in A_0^k}f_{\X}(\bar{b})\\
&=\sum_{s}\X(s)\cdot\sum_{s}\X(s),
\end{align*}
and 
\begin{align*}
[i_1]^{\mA}_s &= [\SUM_x(i_0,x=x)]^{\mA}_s=\sum_{a\in A_0}[i_0]^{\mA}_{s(a/x)}=\sum_{a\in A_0}\sum_{\bar{b}\in A_0^k}f_{\X}(\bar{b})=|A_0|\cdot \sum_{s}\X(s).
\end{align*}
Now $\mA\models\psi$ if and only if $\sum_{s}\X(s)\cdot\sum_{s}\X(s)\leq |A_0|\cdot \sum_{s}\X(s)$. Since $\X$ is nonempty, we have $\sum_{s}\X(s)>0$, and therefore $\mA\models\psi$ iff $\sum_{s}\X(s)\leq |A_0|$.

Let $\X$ and $A_0$ be such that $\sum_{s}\X(s)>|A_0|$. Then $\mA\not\models\psi$, which implies that $\mA_0\not\models_{\X}\phi_{\psi}$. By Proposition \ref{distrprop}, we have $\mA_0\not\models_{\distr(\X)}\phi_{\psi}$. Let $\mA'=(\mA_0,\mA_1,\{f_{\distr(\X)}\})$. Then also $\mA'\not\models\psi$. But now $\sum_{s}\distr(\X)(s)=1\leq |A_0|$, which is a contradiction.
\end{proof}

\subsection{Equi-expressivity of $\FOPTinc$ and $\FOrsum$}\label{equi}

In this subsection, we show that the logics $\FOPTinc$ and $\FOrsum$ are equi-expressive on $\mathbb{R}_{\geq 0}$-structures. The first part, the translation from $\FOPTinc$ to $\FOrsum$, almost already follows from the result of the previous subsection:
\begin{theorem}\label{team2fo_frag}
Let $\phi(v_1,\dots ,v_k)$ be any $\FOPTinc[\tau_0]$-formula and $f$ a $k$-ary function symbol. Then there exists an $\FOrsum[\tau_0\cup\{\leq\}\cup\{f\}]$-sentence $\psi_{\phi}(f)$ such that for any $\mathbb{R}_{\geq 0}$-structure $\mA=(\mA_0,\mA_1,\{f_{\X}\})$ and any probabilistic team $\X$ over $\{v_1,\dots ,v_k\}$
\[
\mA_0\models_{\X}\phi(\bar{v})\iff \mA\models\psi_{\phi}(f),
\]
where $f_{\X}\colon A_0^k\to\mathbb{R}_{\geq 0}$ is a function such that $f_{\X}(s(\bar{v}))=\X(s)$ for all $s\in X$.
\end{theorem}
\begin{proof}
It suffices to complement the proof of Theorem \ref{team2fo} with the case $\theta(\bar{v},\bar{x})=\delta_0(\bar{v},\bar{x})\leq\delta_1(\bar{v},\bar{x})$ since the translations of all subformulas, except for the conditional probability inequality,
are $\FOrsum[\tau_0\cup\{\leq\}\cup\{f\}]$-sentences.

Suppose that $\theta(\bar{v},\bar{x})=\delta_0(\bar{v},\bar{x})\leq\delta_1(\bar{v},\bar{x})$. Then let 
\[
\psi_{\theta}(f,\bar{x}):=\SUM_{\bar{u}}(f(\bar{u}),\delta_0(\bar{u}/\bar{v},\bar{x}))\leq \SUM_{\bar{u}}(f(\bar{u}),\delta_1(\bar{u}/\bar{v},\bar{x})).
\]

Now
\begin{align*}
\mA_0\models_{\X(\bar{a}/\bar{x})}\delta_0\leq\delta_1
\iff &\mA_0\models_{\X}(\delta_0\leq\delta_1)_{(\bar{a}/\bar{x})}\\
\iff &\sum_{s\in S_0}\X(s)\leq\sum_{s\in S_1}\X(s),\text{ where }S_i=\{s\in X\mid \mA_0\models_{s}{\delta_i}_{(\bar{a}/\bar{x})}\}\text{ for } i=0,1\\
\iff &\sum_{\bar{b}\in B_0}f_{\X}(\bar{b})\leq\sum_{\bar{b}\in B_1}f_{\X}(\bar{b}),\text{ where } B_i=\{\bar{b}\in A_0^k\mid \mA_0\models\delta_i(\bar{b}/\bar{v},\bar{a}/\bar{x})\}\text{ for } i=0,1\\
\iff &\mA\models(\SUM_{\bar{u}}(f(\bar{u}),\delta_0(\bar{u}/\bar{v},\bar{x}))\leq \SUM_{\bar{u}}(f(\bar{u}),\delta_1(\bar{u}/\bar{v},\bar{x})))(\bar{a}/\bar{x}).
\end{align*} 
\end{proof}

For the second part, the translation from $\FOrsum$ to $\FOPTinc$, we need the following lemma:
\begin{lem}\label{aggregate_sum_lem}
Every aggregate sum term of the logic $\FOrsum$ can be expressed by a term of the form
\[
\SUM_{\bar{u}}(f(\bar{u}),\delta(\bar{u},\bar{x})),
\]
where $\bar{u}=(u_1,\dots ,u_k)$, and $\delta$ is a disjunction-free and quantifier-free formula, i.e. $\delta\ddfn\lambda\mid\neg\delta\mid\delta\wedge\delta$. 
\end{lem}
\begin{proof}
Consider an aggregate sum of the form 
\[
\SUM_{\bar{u}_0}(f(\bar{u}_0\bar{x}_0),\gamma(\bar{u}_0,\bar{x})),
\]
where $\bar{x}_0$ are among $\bar{x}$, and $\gamma$ may contain disjunctions. The sum can be expressed by the term
\[
\SUM_{\bar{u}_0\bar{u}_1}(f(\bar{u}_0\bar{u}_1),(\gamma^*(\bar{u}_0,\bar{x})\wedge \bar{u}_1=\bar{x}_0)),
\]
where $\gamma^*$ is the formula obtained from $\gamma$ by expressing each disjunction with negation and conjunction in the usual way, i.e. for example, formula $\gamma_0\lor\gamma_1$ is expressed as $\neg(\neg\gamma_0\wedge\neg\gamma_1)$.

To see this, notice that
\[
[\SUM_{\bar{u}_0}(f(\bar{u}_0\bar{x}_0),\gamma(\bar{u}_0,\bar{x}))]_s^{\mA}=\sum_{\bar{a}_0\in B_0}f^{\mA}(s(\bar{a}_0/\bar{u}_0)(\bar{u}_0\bar{x}_0)),
\]
where $B_{0}=\{\bar{a}_0\in A_0^{|\bar{u}_0|}\mid\mA_0\models_s\gamma(\bar{a}_0/\bar{u}_0)\}$, and
\[
[\SUM_{\bar{u}_0\bar{u}_1}(f(\bar{u}_0\bar{u}_1),\gamma^*\wedge\bar{u}_1=\bar{x}_0)]_s^{\mA}=\sum_{\bar{a}_0\bar{a}_1\in B_{01}}f^{\mA}(s(\bar{a}_0\bar{a}_1/\bar{u}_0\bar{u}_1)(\bar{u}_0\bar{u}_1)),
\]
where $B_{01}=\{\bar{a}_0\bar{a}_1\in A_0^{|\bar{u}_0\bar{u}_1|}\mid\mA_0\models_s(\gamma^*\wedge\bar{u}_1=\bar{x}_0)(\bar{a}_0\bar{a}_1/\bar{u}_0\bar{u}_1)\}$. We then have
\[
B_{01}=\{\bar{a}_0s(\bar{x}_0)\in A_0^{|\bar{u}_0\bar{u}_1|}\mid\mA_0\models_s\gamma(\bar{a}_0/\bar{u}_0)\},
\]
from which it follows that
\begin{align*}
[\SUM_{\bar{u}_0\bar{u}_1}(f(\bar{u}_0\bar{u}_1),\gamma^*\wedge\bar{u}_1=\bar{x}_0)]_s^{\mA}
&=\sum_{\bar{a}_0\bar{a}_1\in B_{01}}f^{\mA}(s(\bar{a}_0\bar{a}_1/\bar{u}_0\bar{u}_1)(\bar{u}_0\bar{u}_1))\\
&=\sum_{\bar{a}_0s(\bar{x}_0)\in B_{01}}f^{\mA}(s(\bar{a}_0s(\bar{x}_0)/\bar{u}_0\bar{u}_1)(\bar{u}_0\bar{u}_1))\\
&=\sum_{\bar{a}_0\in B_{0}}f^{\mA}(s(\bar{a}_0/\bar{u}_0)(\bar{u}_0\bar{x}_0))\\
&=[\SUM_{\bar{u}_0}(f(\bar{u}_0\bar{x}_0),\gamma)]_s^{\mA}.
\end{align*}
\end{proof}
\begin{theorem}\label{FO2team}
Let $\psi(f)$ be any $\FOrsum[\tau_0\cup\{\leq\}\cup\{f\}]$-sentence, where $f$ is a $k$-ary function symbol. Then there exists an $\FOPTinc[\tau_0]$-formula $\phi_{\psi}(v_1,\dots ,v_k)$ such that for any $\mathbb{R}_{\geq 0}$-structure $\mA=(\mA_0,\mA_1,\{f_{\X}\})$ and any nonempty  probabilistic team $\X$ over $\{v_1,\dots ,v_k\}$ 
\[
\mA_0\models_{\X}\phi_{\psi}(\bar{v})\iff \mA\models\psi(f),
\]
where $f_{\X}\colon A_0^k\to\mathbb{R}_{\geq 0}$ is a function such that $f_{\X}(s(\bar{v}))=\X(s)$ for all $s\in X$. 
\end{theorem}
\begin{proof}
Without loss of generality, we may assume that $\psi(f)$ is in prenex normal form, i.e.
\[
\psi(f)=Q_1x_1\dots Q_nx_n\theta(f,\bar{x}),
\]
where $Q_i\in\{\exists,\forall\}$, $1\leq i\leq n$, and $\theta$ is quantifier free.

We then let $\phi_{\psi}(\bar{v}):=Q_1^1x_1\dots Q_n^1x_n\phi_{\theta}(\bar{v},\bar{x})$, where $\phi_{\theta}(\bar{v},\bar{x})$ is defined inductively as follows:
\begin{itemize}
\item[(1)] Suppose that $\theta(\bar{x})=\lambda(\bar{x})$, where $\lambda$ is a first-order atomic formula ($f$ does not appear in $\lambda$). Then let $\phi_{\theta}(\bar{v},\bar{x}):=\lambda(\bar{x})$.

Now
\begin{align*}
\mA_0\models_{\X(\bar{a}/\bar{x})}\lambda(\bar{x})
\iff &\mA_0\models_{\X}\lambda(\bar{x})_{(\bar{a}/\bar{x})} \quad(\text{by Prop. \ref{lemmaprop}})\\
\iff &\mA_0\models_{s}\lambda(\bar{x})_{(\bar{a}/\bar{x})}\text{ for all } s\in \supp(X)\\
\iff &\mA\models\lambda(\bar{a}/\bar{x})\quad \text{(}f \text{ does not appear in } \lambda\text{).}
\end{align*}

\item[(2)] By Lemma \ref{aggregate_sum_lem}, it suffices to consider the case 
\[
\theta(f,\bar{x})=\SUM_{\bar{u}}(f(\bar{u}),\delta_0(\bar{u},\bar{x}))\leq \SUM_{\bar{u}}(f(\bar{u}),\delta_1(\bar{u},\bar{x})),
\]
where $\delta$ is a disjunction-free and quantifier-free formula. Then let $\phi_{\theta}(\bar{v},\bar{x}):=\delta_0(\bar{v}/\bar{u},\bar{x})\leq \delta_1(\bar{v}/\bar{u},\bar{x})$. This is similar to the proof of Theorem \ref{team2fo_frag}.

\item[(3)] Suppose that $\theta(f,\bar{x})=\neg\theta_0(f,\bar{x})$. Then let $\phi_{\theta}(\bar{v},\bar{x}):=\wcn\phi_{\theta_0}(\bar{v},\bar{x})$.
Now
\begin{align*}
\mA_0\models_{\X(\bar{a}/\bar{x})}\wcn\phi_{\theta_0}
\iff &\mA_0\not\models_{\X(\bar{a}/\bar{x})}\phi_{\theta_0}(\bar{v},\bar{x})\quad\text{(} \X(\bar{a}/\bar{x}))\text{ is nonempty)}\\
\iff &\mA\not\models\theta_0(f,\bar{x})(\bar{a}/\bar{x}) \quad\text{(by the induction hypothesis)}\\
\iff &\mA\models\neg\theta_0(f,\bar{x})(\bar{a}/\bar{x}).
\end{align*}

\item[(4)] Suppose that $\theta(\bar{x})=\theta_0(\bar{x})\wedge\theta_0(\bar{x})$. Then let $\phi_{\theta}(\bar{v},\bar{x}):=\phi_{\theta_0}(\bar{v},\bar{x})\wedge\phi_{\theta_1}(\bar{v},\bar{x})$. The claim directly follows from the induction hypothesis.

\item[(5)] Suppose that $\theta(\bar{x})=\theta_0(\bar{x})\lor\theta_0(\bar{x})$. Then let $\phi_{\theta}(\bar{v},\bar{x}):=\phi_{\theta_0}(\bar{v},\bar{x})\vvee\phi_{\theta_1}(\bar{v},\bar{x})$. The claim directly follows from the induction hypothesis.

\end{itemize}
Now
\begin{align*}
\mA_0\models_{\X}Q_1^1x_1\dots Q_n^1x_n\phi_{\theta}(\bar{v},\bar{x})
&\iff Q_1a_1,\dots ,Q_na_n\in A_0, \text{ } \mA_0\models_{\X(\bar{a}/\bar{x})}\phi_{\theta}(\bar{v},\bar{x})\\
&\iff Q_1a_1,\dots ,Q_na_n\in A_0, \text{ } \mA\models\theta(\bar{a}/\bar{x})\\
&\iff \mA\models Q_1x_1\dots Q_nx_n\theta(\bar{x}).
\end{align*}
\end{proof}

By combining Theorems \ref{team2fo_frag} and \ref{FO2team}, we obtain that $\FOPTinc$ and $\FOrsum$ are equi-expressive on $\mathbb{R}_{\geq 0}$-structures.

\section{Translation from $\FOrxsum$ to $\logicFont{FFP}_{\mathbb{R}}$}\label{FFP}

In this section, we present a translation from $\FOrxsum$ to  a fragment of $\logicFont{FFP}_{\mathbb{R}}$. The logic $\logicFont{FFP}_{\mathbb{R}}$ was introduced in \cite{gradel95} as a logic for PTIME over reals  (w.r.t. ordered structures). It is a fixed point logic with constants for every real number. In the fragment that we consider, the constants are restricted to 0 and 1, and therefore the data complexity of the fragment corresponds to the class $\textsc{P}_{\mathbb{R}}^0$, i.e., the class of languages over $\mathbb{R}$ decidable in polynomial time by a BSS-machine with restriction to machine constants 0 and 1. The translation gives us an upper bound for the data complexity of $\FOrxsum$. We summarize those definitions from \cite{gradel95} which are needed for the translation; for further details on $\logicFont{FFP}_{\mathbb{R}}$, see \cite{gradel95}.

A two-sorted structure $\mA=(\mA_0,\mA_1,F)$ is called an $\mathbb{R}$-structure if
\[
\mA_1=\mathcal{R}:=(\mathbb{R},+,-,\times,/,\text{sign},=,<,0,1).
\]
We also denote $\tau_{\mathcal{R}}=\{+,-,\times,/,\text{sign},=,<,0,1\}$. In the following, we restrict to functional $\mathbb{R}$-structures or $\mathbb{R}$-\textit{algebras}. These are $\mathbb{R}$-structures $(\mA_0,\mathcal{R},F)$ such that structure $\mA_0$ is a plain set $A_0$, i.e. $\tau_0=\varnothing$. 

We consider a fragment of the \textit{functional fixed-point logic for} $\mathbb{R}$-\textit{algebras}, or  $\logicFont{FFP}_{\mathbb{R}}$. First-order $\tau_0$-terms are defined in the usual way. Note that since $\tau_0=\varnothing$, we only have variables as first-order terms. The fragment of $\logicFont{FFP}_{\mathbb{R}}$ over a vocabulary $\tau_0\cup\tau_{\mathcal{R}}\cup\sigma=\tau_{\mathcal{R}}\cup\sigma$ is the set of numerical terms, defined as follows:
\[
i\ddfn c\mid f(\bar{x})\mid i+i\mid i-i\mid i\times i\mid i/i\mid \text{sgn}(i)\mid \max_{\bar{x}}i(\bar{y})\mid \textbf{fp}[Z(\bar{z})\leftarrow i(Z,\bar{z})](\bar{y})
\]
where $c\in\{0,1\}$, $f$ and $Z$ are function symbols such that $f\in\sigma$ and $Z\notin\sigma$, $\bar{x},\bar{y},\bar{z}$ are tuples of distinct variables with $|\bar{x}|=\ar(f)$, $\Var(\bar{x})\subseteq\Var(\bar{y})$, and $|\bar{y}|=|\bar{z}|=\ar(Z)$. 

First-order terms are interpreted in the usual way. Intended interpretations for most of the numerical terms are clear. We give interpretations for the non-obvious ones: $\text{sgn}(i)$, $\max_{\bar{x}}i(\bar{y})$, and $\textbf{fp}[Z(\bar{z})\leftarrow i(Z,\bar{z})](\bar{y})$. We define
\begin{align*}
[\text{sgn}(i)]_s^{\mA}:=
\begin{cases}
1, \quad &\text{ when } [i]_s^{\mA}>0\\
0, \quad &\text{ when } [i]_s^{\mA}=0\\
-1, \quad &\text{ when } [i]_s^{\mA}<0,
\end{cases}
\end{align*}
and
\begin{align*}
[\max_{\bar{x}}i(\bar{y})]_s^{\mA}:=\max\{[i(\bar{y})]_{s(\bar{a}/\bar{x})}^{\mA} \mid \bar{a}\in A_0^{|\bar{x}|}\}.
\end{align*}
Because of the terms of the form $\textbf{fp}[Z(\bar{z})\leftarrow i(Z,\bar{z})](\bar{y})$,  we also allow partially defined functions $Z$ that map tuples from $A_0$ to $\mathbb{R}$. We define a \textit{partial} $\mathbb{R}$-\textit{algebra} as an $\mathbb{R}\cup\{\text{undef}\}$-algebra obtained by extending the basic operations on $\mathbb{R}$ as follows: 
if $[j]_s^{\mA}=\text{undef}$, then
\[
[i+j]_s^{\mA}=[i-j]_s^{\mA}=\text{undef},\qquad [\text{sign}(j)]_s^{\mA}=\text{undef},
\]
and
\[[i\times j]_s^{\mA}=[i/j]_s^{\mA}=
\begin{cases}
0, \quad &\text{when } [i]_s^{\mA}=0\\
\text{undef}, \quad &\text{when } [i]_s^{\mA}\neq 0.
\end{cases}
\]
Additionally, $[\max_{\bar{x}}i(\bar{y})]_s^{\mA}=\text{undef}$, when $[i(\bar{y})]_{s(\bar{a}/\bar{x})}^{\mA}=\text{undef}$ for some $\bar{a}\in A_0^{|\bar{x}|}$.

Let $i(Z,\bar{z})$ be a numerical term of vocabulary $\tau_{\mathcal{R}}\cup\{Z\}$. We write $[i(Z,\bar{z})]_s^{\mA,Z}$ for the interpretation of the term $i(Z,\bar{z})$ in the structure obtained from $\mA$ by adding a suitable partial function $Z\colon A_0^{\ar(Z)}\to\mathbb{R}$. The term $i(Z,\bar{z})$ induces an operator $F_i^{\mA}$ that updates partially defined functions $Z$ as follows:
\[
F_i^{\mA}Z(s(\bar{z}))=
\begin{cases}
[i(Z,\bar{z})]_s^{\mA,Z}, \quad &\text{when } Z(s(\bar{z}))=\text{undef}\\
Z(s(\bar{z})), \quad &\text{otherwise.}
\end{cases}   
\]
This defines a sequence of partial functions $Z^j\colon A_0^{\ar(Z)}\to\mathbb{R}$ such that
\begin{align*}
&Z^0(\bar{a})=\text{undef} \quad\text{ for all } \bar{a}\in A_0^{\ar(Z)}\\
&Z^{j+1}=F_i^{\mA}Z^{j}.
\end{align*}
Note that $Z^{j+1}=Z^j$ for some $j\leq |A_0|^{\ar(Z)}$, and after this $j$, any further iterations do not update the function. We call this $Z^j$ the \textit{fixed point of} $F_i^{\mA}$. We define
\[
[\textbf{fp}[Z(\bar{z})\leftarrow i(Z,\bar{z})](\bar{y})]_s^{\mA}=Z^{\infty}(s(\bar{y}))
\]
where $Z^{\infty}$ is the fixed point of $F_i^{\mA}Z$.

A function $E\colon A_0\to\mathbb{R}$ that is a bijection from $A_0$ to $\{0,\dots,|A_0|-1\}$ is called a \textit{ranking}. We say that a structure $\mA$ is \textit{ranked} if the set $F$ contains a ranking. A given ranking $E$ induces a ranking $E_k$ of $k$-tuples for any $k>0$. The ranking $E_k$ is definable, and we will use the abbreviation $\underline{x}$ for $E_k(\bar{x})$ where $\bar{x}$ is a $k$-tuple of first-order variables.

Let $\tau_0$ be a finite relational vocabulary, and $\mA_0$ a finite $\tau_0$-structure. We define the structure $\mA_0^*$ as the plain set $A_0$. We can make an $\mathbb{R}$-algebra $\mA^*=(\mA_0^*,\mathcal{R},F)$ of vocabulary $\tau_{\mathcal{R}}\cup\sigma$ by adding to $\sigma$ characteristic functions $\chi_R$ for all relation symbols $R\in\tau_0$. Let $\phi$ be a first-order formula of vocabulary $\tau_0$. Then the characteristic function of $\phi$, denoted by $\chi[\phi]$, is definable in $\logicFont{FFP}_{\mathbb{R}}[\tau_{\mathcal{R}}\cup\sigma]$. Moreover, if $i,j$ are numerical $\tau_\mathcal{R}\cup\sigma$-terms, then functions $\chi[i=j]$ and $\chi[i\leq j]$ are also definable in $\logicFont{FFP}_{\mathbb{R}}[\tau_{\mathcal{R}}\cup\sigma]$. (See \cite{gradel95} or the proof of Theorem \ref{toFFP} below.)

The next theorem shows that $\FOrxsum[\tau_0\cup\{\leq\}\cup\{f\}]$-formulas can be viewed as functions of $\logicFont{FFP}_{\mathbb{R}}$. Note that the corresponding $\logicFont{FFP}_{\mathbb{R}}$-term will be over $\tau_{\mathcal{R}}\cup\sigma$, a different vocabulary since in $\mathbb{R}$-algebras $\mA^*$ each relation $R^{\mA}\subseteq A_0^{\ar(R)}$ is replaced with its characteristic function $\chi_R\colon A_0^{\ar(R)}\to\mathbb{R}$.

\begin{theorem}\label{toFFP}
Let $\phi$ be any $\FOrxsum[\tau_0\cup\{\leq\}\cup\{f\}]$-formula, and let $\sigma$ be a vocabulary that contains function symbols $E$ and $f$, as well as $\chi_R$ for all relation symbols $R\in\tau_0$. Then there exists an $\logicFont{FFP}_{\mathbb{R}}[\tau_{\mathcal{R}}\cup\sigma]$-term $i_{\phi}$ such that for any $\mathbb{R}_{\geq 0}$-structure $\mA=(\mA_0,\mA_1,\{f^{\mA}\})$ and any assignment $s$
\[
\mA\models_s\phi \iff [i_{\phi}]_s^{\mA^*}=1
\]
where $\mA^*=(\mA_0^*,\mathcal{R},F)$ is an $\mathbb{R}$-algebra such that structure $\mA_0^*$ is the plain set $A_0$, and $F$ contains a ranking $E$, the function $f^{\mA}$, and the characteristic functions $\chi_R$ for all relations $R\in\tau_0$.
\end{theorem}
\begin{proof}
We begin by showing how to translate any numerical $\FOrxsum$-term $i$ of vocabulary $\tau_0\cup\{\leq\}\cup\{f\}$. We denote by $i^*$ the translation which is a numerical $\logicFont{FFP}_{\mathbb{R}}$-term of vocabulary $\tau_{\mathcal{R}}\cup\sigma$.
\begin{itemize}
\item[(1)] If $i=f(\bar{x})$, then $i^*:=f(\bar{x})$.
\item[(2)] If $i=i_0\times i_1$, then $i^*:=i_0^*\times i_1^*$.
\item[(3)] If $i=\SUM_{\bar{x}}(i_0(\bar{y}),\gamma(\bar{y}))$ where $\Var(\bar{x})\subseteq\Var(\bar{y})$, then 
\[
i^*:=\max_{\bar{x}}\textbf{fp}[Z(\bar{y})\leftarrow j(Z,\bar{y})](\bar{y}),
\]
where 
\begin{align*}
j(Z,\bar{y})=&\chi[\underline{x}=0]\times i_0^*(\bar{y})\times\chi[\gamma(\bar{y})]\\
&+\max_{\bar{u}}\left(\chi[\underline{x}=\underline{u}+1]\times(Z(\bar{y}(\bar{u}/\bar{x}))+i_0^*(\bar{y})\times\chi[\gamma(\bar{y})])\right).
\end{align*} 
(In the above, $\bar{y}(\bar{u}/\bar{x})$ denotes the tuple obtained from $\bar{y}$ by replacing $\bar{x}$ with $\bar{u}$.)
\end{itemize}
We continue by defining the corresponding $\logicFont{FFP}_{\mathbb{R}}[\tau_{\mathcal{R}}\cup\sigma]$-terms for formulas $\phi$.
\begin{itemize}
\item[(4)] Let $\phi=\lambda$, where $\lambda$ is a first-order atomic formula of vocabulary $\tau_0$. Then $\lambda=R(\bar{x})$ for some $R\in\tau_0$. Now, we let $i_{\lambda}:=\chi_R(\bar{x})$. (Note that $R$ may be the equality relation, so this also covers the case $\lambda=x_0=x_1$.)
\item[(5)] If $\phi=i_0\leq i_1$, then 
\begin{align*}
i_{\phi}&:=\chi[i_0^*= i_1^*\vee i_0^*< i_1^*]\\
&=\chi[i_0^*= i_1^*]+\chi[i_0^*< i_1^*]-\chi[i_0^*= i_1^*]\times\chi[i_0^*< i_1^*]
\end{align*}
where 
\[
\chi[i_0^*= i_1^*]=1-[\text{sign}(i_0^*- i_1^*)]^2
\]
and
\[
\chi[i_0^*< i_1^*]=([\text{sign}(i_1^*- i_0^*)]^2+\text{sign}(i_1^*- i_0^*))/2.
\]
\item[(6)] If $\phi=\neg\theta_0$, then $i_{\phi}:=1-i_{\theta_0}$.
\item[(7)] If $\phi=\theta_0\wedge\theta_1$, then $i_{\phi}:=i_{\theta_0}\times i_{\theta_1}$.
\item[(8)] If $\phi=\theta_0\vee\theta_1$, then $i_{\phi}:=i_{\theta_0}+ i_{\theta_1}-i_{\theta_0}\times i_{\theta_1}$.
\item[(9)] If $\phi=\exists x\theta_0$, then $i_{\phi}:=\max_x i_{\theta_0}$.
\item[(10)] If $\phi=\forall x\theta_0$, then $i_{\phi}:=1-\max_x (1-i_{\theta_0})$.
\end{itemize}
\end{proof}

\section{Conclusion}
 We have defined new tractable logics for the framework of probabilistic team semantics that generalize the recently defined logic  $\FOT$  that is expressively complete for first-order team properties. Our logics employ new probabilistic atoms that resemble so-called extended atoms from the team semantics literature. We also defined counterparts of our logics over metafinite structures and showed that all of our logics can be translated into functional fixed point logic giving a deterministic polynomial-time upper bound for data complexity with respect to BSS-computations. 
 
 The following questions remain open:

\begin{itemize}
\item What is the exact data complexity  of our logics in the BSS-model?
\item Is it possible to axiomatize  (fragments) of our new logics?
\end{itemize}
Note that by Theorem \ref{val} the logic $\FOPTcondineq$ cannot be fully axiomatized but, e.g.,  several axiomatizations are know for mere probabilistic independence atoms (see \cite{CoranderHKPV16} for references).


\medskip

\bibliographystyle{plain}
\bibliography{biblio1,biblio2,biblio3}

\end{document}